\documentclass[11pt,letterpaper]{article}

\usepackage{fixltx2e}
\usepackage{fullpage}
\usepackage{amsmath}
\usepackage{amssymb}
\usepackage{amsthm}
\usepackage{enumerate}
\usepackage[latin1]{inputenc}
 \newcommand{\href}[2]{#2}
\usepackage{algorithmic} 
\usepackage[ruled,vlined,algo2e,linesnumbered]{algorithm2e}
\DontPrintSemicolon
\SetNlSty{}{}{:}
\SetCommentSty{textit}
\SetVlineSkip{0pt}

\newtheorem{theorem}{Theorem}
\newtheorem{proposition}[theorem]{Proposition}
\newtheorem{corollary}[theorem]{Corollary}
\newtheorem{lemma}[theorem]{Lemma}

\newtheorem{definition}[theorem]{Definition}

\newtheorem{remark}[theorem]{Remark}

\newcommand{\NN}{\ensuremath{\mathbb{N}}}
\newcommand{\ZZ}{\ensuremath{\mathbb{Z}}}
\newcommand{\QQ}{\ensuremath{\mathbb{Q}}}
\newcommand{\RR}{\ensuremath{\mathbb{R}}}
\newcommand{\CC}{\ensuremath{\mathbb{C}}}

\newcommand{\abs}[1]{\ensuremath{\mathopen\lvert #1 \mathclose\rvert}}
\newcommand{\norm}[1]{\ensuremath{\mathopen\lVert #1 \mathclose\rVert}}

\newcommand{\height}[1]{\ensuremath{\abs{#1}_\infty}}

\newcommand{\sgn}{\mathop{\mathrm{sgn}}}
\newcommand{\val}{\mathop{\mathrm{val}}}
\newcommand{\cont}{\mathop{\mathrm{cont}}}
\newcommand{\diag}{\mathop{\mathrm{diag}}}

\newcommand{\dotleq}{\ensuremath{\preccurlyeq}}
\newcommand{\dotgeq}{\ensuremath{\succcurlyeq}}

\newcommand{\ball}{\ensuremath{\mathrm{B}}}

\newcommand{\bit}{\ensuremath{\operatorname{bit}}}
\newcommand{\sep}{\ensuremath{\operatorname{sep}}}
\newcommand{\bitsize}[1]{\ensuremath{\mathcal{L}\left( #1 \right)}\xspace}

\newcommand{\x}{\mathbf{x}\xspace}
\newcommand{\dg}[2][]{\ensuremath{\mathrm{deg}_{#1}(#2)}\xspace}
\def\ceil#1{\lceil #1 \rceil}

\newcommand{\transpose}{\ensuremath{\mathsf{T}}}

\usepackage{enumerate}
\def\floor#1{\lfloor #1 \rfloor}
\def\Paren#1{\left( #1 \right)}

\title{Exact Algorithms for Solving Stochastic Games\thanks{An extended abstract of this paper was presented at STOC'11. Hansen, Miltersen and Tsigaridas acknowledge support from
the Danish National Research Foundation and The National Science Foundation of
China (under the grant 61061130540) for the Sino-Danish Center for the Theory of
Interactive Computation, within which this work was performed. They also
acknowledge support from the Center for Research in Foundations of Electronic Markets (CFEM), supported by the Danish Strategic Research Council.}}

\author{Kristoffer Arnsfelt Hansen%
\thanks{Computer Science Department, Aarhus University, Denmark.}
\and
Michal Kouck\'{y}%
\thanks{Institute of Mathematics of Czech Academy of Sciences.
Partially supported by GA \v{C}R P202/10/0854, project No.~1M0021620808 of M\v{S}MT \v{C}R, 
    Institutional Research Plan No.~AV0Z10190503 and grant IAA100190902 of GA AV \v{C}R.}
\and
Niels Lauritzen%
\thanks{Mathematics Department, Aarhus University, Denmark.}
\and
Peter Bro Miltersen%
\thanks{Computer Science Department, Aarhus University, Denmark.}
{\ }and{\ } 
Elias P. Tsigaridas\thanks{
Computer Science Department, Aarhus University, Denmark.
Partially supported by an individual postdoctoral grant from the
Danish Agency for Science, Technology and Innovation.}
}

\date{\today}

\begin{document}
\maketitle

\begin{abstract}
  Shapley's {\em discounted stochastic games}, Everett's {\em recursive games} and Gillette's {\em undiscounted stochastic games}
  are classical models of game theory describing two-player zero-sum
  games of potentially infinite duration. We describe
  algorithms for exactly solving these games.  
When the number of positions of the game is
  constant, our algorithms run in polynomial time.
\end{abstract}




\section{Introduction}
Shapley's model of finite {\em stochastic games} \cite{PNAS:Shapley53}
is a classical model of game theory describing two-player zero-sum
games of (potentially) infinite duration. Such a game 
is given by a finite set of positions $1, \ldots, N$, with a $m_k \times
n_k$ {\em reward}
matrix $(a_{ij}^k)$ associated to each position $k$, and an $m_k \times n_k$ {\em
  transition} matrix $(p_{ij}^{kl})$ associated to each pair of positions
$k$ and $l$. The game is played in rounds, with some position $k$ being
the {\em current} position in each round. At each such round, Player I chooses an
action $i \in \{1,2,\ldots,m_k\}$ while simultaneously, Player II
chooses an action $j \in \{1,2,\ldots,n_k\}$, after which the (possibly negative) {\em reward}
$a^k_{ij}$ is paid by Player II to Player I, and with probability
$p_{ij}^{kl}$ the current position becomes $l$ for the next round.

During play of a stochastic game, a sequence of rewards is paid by Player II
to Player I. There are three standard ways of associating a {\em payoff} to
Player I from such a sequence, leading to three different variants of
the stochastic game model:

{\it Shapley games. }
  In Shapley's original paper, the payoff is simply the sum of
  rewards. While this is not well-defined in general, in Shapley's
  setting it is required that for all positions
  $k$, $\sum_l p^{kl}_{ij} < 1$, with the remaining probability mass
  resulting in termination of play. Thus, no matter which actions are
  chosen by the players, play eventually ends with probability 1, making the
  payoff well-defined except with probability 0. We shall refer to
  this original variant of the stochastic games model as {\em Shapley
    games}. Shapley observed that an alternative formulation of this
  payoff criterion is to require $\sum_l p^{kl}_{ij} =
  1$, but {\em discounting} rewards, i.e., penalizing a reward
  accumulated at time $t$ by a factor of $\gamma^t$ where $\gamma$ is a {\em
    discount factor} strictly between 0 and 1. Therefore, Shapley
  games are also often referred to as {\em discounted} stochastic games. Using the Banach fixed point
  theorem in combination with the von Neumann minimax theorem for
  matrix games, Shapley showed that all Shapley games have a value,
  or, more precisely, a {\em value vector}, one value for each position.
Also, the values can be guaranteed by both players by a {\em stationary
  strategy}, i.e., a strategy that associates a fixed probability
distribution on actions to each position and therefore does not take
history of play into account.

{\it Gillette games. } Gillette \cite{gil} requires that for all
$k,i,j$, $\sum_l p^{kl}_{ij} = 1$, i.e., all plays are infinite. The
total payoff to Player I is $\liminf_{T \rightarrow \infty}
(\sum_{t=1}^T r_i)/T$ where $r_t$ is the reward collected at round
$t$. Such games are called {\em undiscounted} or {\em limiting
  average} stochastic games. In this paper, for coherence of
terminology, we shall refer to them as {\em Gillette games}. It is
much harder to see that Gillette games have values than that Shapley
games do. In fact, it was open for many years if the concrete game
{\em The Big Match} with only three positions that was suggested by
Gillette has a value. This problem was resolved by Blackwell and
Ferguson \cite{Bigmatch}, and later, Mertens and Neyman \cite{MN}
proved in an ingenious way that all Gillette games have value vectors,
using the result of Bewley and
Kohlberg~\cite{BewleyKohlberg1976}. However, the values can in general
only be approximated arbitrarily well by strategies of the players,
not guaranteed exactly, and non-stationary strategies (taking history
of play into account) are needed even to achieve such
approximations. In fact, {\em The Big Match} proves both of these
points.

{\it Everett games. }
Of generality between Shapley games and Gillette games is
  the model of {\em recursive games} of Everett \cite{AMS:Everett57}.
We shall refer to these games as {\em Everett games}, also to avoid
confusion with the largely unrelated notion of recursive games of
Etessami and Yannakakis \cite{EY}.
In Everett's model, we have $a^k_{ij} = 0$ for all $i,j,k$, i..e,
rewards are {\em not} accumulated during play. For each particular
$k$, we can have either $\sum_l p^{kl}_{ij} < 1$ or $\sum_l p^{kl}_{ij} =
1$. In the former case, a prespecified payoff $b_{ij}^k$ is associated to
the termination outcome. Payoff 0 is associated with infinite play.
The special case of Everett games where $b^k_{ij} = 1$ for all
$k,i,j$ 
has been studied under the name of {\em concurrent reachability
  games} 
in the computer
science literature \cite{AHK,
  ChatQest,LICS:HKM09,CRGIteration}. 
Everett showed that Shapley games can be seen 
as a special case of Everett games. Also, it is easy to see 
Everett games as a special case of Gillette games. It
was shown in Everett's original paper that all Everett games have 
value vectors. Like Gillette games, the values can in general only
be approximated arbitrarily well, but unlike Gillette games, stationary strategies are
sufficient for guaranteeing such approximations. 

For formal definitions and proofs of some of the facts above, see
Section~\ref{sec-pre}.

\subsubsection*{Our Results}

In this paper we consider the
problem of exactly solving Shapley, Everett and Gillette games, 
i.e., computing the value of a
given game.  The variants of these two problems for the case of {\em
  perfect information} (a.k.a. {\em turn-based}) games are
well-studied by the computer science community, but not known to be
polynomial time solvable: The tasks of solving perfect information
Shapley, Everett and Gillette games and the task of solving
Condon's {\em simple stochastic games} \cite{Condon92} are
polynomial time equivalent \cite{AnMi09}. 
Solving simple stochastic
games in polynomial time is by now a famous open problem.  
As we consider algorithms for the more
general case of imperfect information games, we,
unsurprisingly, do not come up with polynomial time algorithms.
However, we describe algorithms for all three classes of games that run in
polynomial time when the number of positions is constant and our
algorithms are the first algorithms with this property.  As the values
of all three kinds of games may be irrational but algebraic numbers, 
our algorithms
output real algebraic numbers in {\em isolating interval representation,} i.e.,
as a square-free polynomial with rational coefficients for which the
value is a root, together with an (isolating) interval with rational
endpoints in which this root is the only root of the polynomial. To be precise,
our main theorem is:

\vspace{5pt}
{\bf Theorem.}
{\em For any constant $N$, there is a polynomial time
algorithm that takes as input a Shapley, Everett or Gillette game 
with $N$ positions and outputs its value 
vector using isolating interval encoding.
Also, for the case of a Shapley games, an optimal
stationary strategy for the game in isolating interval encoding 
can be computed in
polynomial time. Finally, for Shapley as well as Everett games,
given an
additional input parameter $\epsilon > 0$, an $\epsilon$-optimal
stationary strategy using only (dyadic) rational valued probabilities can
be computed in time polynomial in the representation of the game
and $\log(1/\epsilon)$.}
\vspace{5pt}

We remark that when the number of positions $N$ is constant, what remains to
vary is (most importantly) the number of actions $m$ for each player 
in each position and (less importantly) the bitsize $\tau$ of transition
probabilities and payoffs.
We also remark that Etessami and Yannakakis \cite{LMCS:EY08} showed that the
  bitsize of the isolating interval encoding of the value of a
  discounted stochastic game as well as the value of a recursive game
  may be exponential in the number of positions of the game and that
  Hansen, Kouck\'{y} and Miltersen \cite{LICS:HKM09} showed that the
  bitsize of an $\epsilon$-optimal strategy for a recursive game using
  binary representation of probabilities may be exponential in the
  number of positions of the game. Thus, merely from the size of the
  output to be produced, there can be no polynomial time algorithm for
  the tasks considered in the theorem 
  without some restriction on $N$.
Nevertheless, the time complexity of our algorithm has
a dependence on $N$ which is very bad and not matching the 
size of the output. For the case of Shapley games, the exponent
in the polynomial time bound is $O(N)^{N^2}$ while for the case
of Everett games and Gillette games, the exponent is $N^{O(N^2)}$. 
Thus, getting a better
dependence on $N$ is a very interesting open problem. 

Prior to our work, algorithms for solving stochastic games
relied either
on generic reductions to decision procedures for the first order
theory of the reals \cite{LMCS:EY08,HE}, or, for the case of 
Shapley games and concurrent reachability games 
on value or 
strategy iteration \cite{RCN,ChatQest}. For all these
algorithms, the complexity
is at least exponential {\em even when the number of positions is a 
constant and even when only a crude approximation is required 
\cite{CRGIteration}}. Nevertheless, as is the case for the 
algorithms based on
reductions to decision procedures for the first order theory of the reals,
our algorithms rely on the theory of semi-algebraic geometry 
\cite{BasuPollackRoy2006}, but in a more indirect way as we explain below.

Our algorithms are based on a simple recursive bisection pattern which
is in fact a very natural and in retrospect unsurprising approach to solving
stochastic games. 
However, in order to set the parameters of the algorithm in a way
that makes it correct, 
we need {\em separation bounds}
for values of stochastic games of given type and parameters; 
lower bounds on the absolute value of games of non-zero value.
Such bounds are obtained by bounding the algebraic degree and
coefficient size of the defining univariate polynomial and applying standard
arguments, so the task at hand boils down to determining as good bounds on degree and coefficient size as possible; with better bounds leading to faster algorithms.  To get these bounds, we apply the general machinery of real algebraic geometry and semi-algebraic geometry following closely the techniques of the seminal work of Basu, Pollack and Roy \cite{BasuPollackRoy2006}. That is, for each of 
the three types of games, we describe how for a given game $G$ to derive a formula in the first order theory of the real numbers uniquely defining the value of $G$. This essentially involves formalizing statements proved by Shapley, Everett, and Mertens and Neyman together with elementary properties of linear programming. Now, we apply the powerful tools of {\em quantifier elimination} \cite[Theorem 14.16]{BasuPollackRoy2006} and {\em sampling} \cite[Theorem 13.11]{BasuPollackRoy2006} to show the appropriate bounds on degree and coefficient size. 
We stress that these procedures are only carried out in our proofs; they are not carried out by our algorithms. Indeed, if they were, the time complexity of the algorithms would be exponential, even for a constant number of positions. 
While powerful, the semi-algebraic approach has the disadvantage of giving rather imprecise bounds. Indeed, as far as we know, all published versions of the quantifier elimination theorem and the sampling theorem have unspecified constants (``big-Os''), leading to unspecified constants in the code of our algorithms.  
Only for the case of Shapley games, are we able to do somewhat better, their mathematics being so simple that we can avoid the use of the general tools of quantifier elimination and sampling and instead base our bounds on solutions to the following very natural concrete problem of real algebraic geometry that can be seen as a very special case of the sampling problem: 

{\em Given a system of $m$ polynomials in $n$ variables (where $m$ is
  in general different from $n$) of degree bounded by $d$, whose
  coefficients have bitsizes at most $\tau$, and an {\em isolated} (in
  the Euclidean topology) real root of the system, what is an upper
  bound on its algebraic degree as a function of $d$ and $n$? What is
  a bound on the bitsizes of the coefficients of the defining
  polynomial?}  

Basu, Pollack and Roy \cite[Corollary 13.18]{BasuPollackRoy2006} stated the upper bound $O(d)^n$ on the
algebraic degree as a corollary of the sampling theorem.  We give a
constructive bound of $(2d+1)^n$ on the algebraic degree and we derive
an explicit bound on the coefficients of the defining polynomial.  
We emphasize that our techniques for doing this are standard in the
context of real algebraic geometry; in particular the deformation
method and u-resultants are used. However, we find it surprising that
(to the best of our knowledge) no explicit constant for the big-O was
previously stated for this very natural problem. Also, we do not
believe that $(2d+1)^n$ is the final answer and would like to see an
improvement. We hope that by stating some explicit bound we will
stimulate work improving it. We note that for the case of 
isolated complex roots, explicit bounds appeared
recently, see Emiris, Mourrain and Tsigaridas \cite{emt-issac-2010} 
and references therein.

The degree bounds for the algebraic problem lead to
upper bounds on the algebraic degree 
of the values of Shapley games as
a function of the combinatorial parameters of the game.
We also provide corresponding lower bounds
by proving that polynomials that have among their real roots the value of certain Shapley games are irreducible. We prove irreducibility based on Hilbert's irreducibility theorem and a generalization of the Eisenstein criterion,
As these bounds may be of independent interest, we state them
explicitly: 

{\em The value of any Shapley game with $N$ positions, $m$
  actions for each player in each position, and rational payoffs and
  transition probabilities, is an algebraic number of degree at most
  $(2m+5)^N$.  Also, for any $N,m \geq 1$ there exists a game with
  these parameters such that its value is an algebraic number of
  degree $m^{N-1}$.}  

The lower bound strengthens a result of Etessami
and Yannakakis \cite{LMCS:EY08} who considered the case of $m=2$ and
proved a $2^{\Omega(N)}$ lower bound.  For the more general case of
Everett games and Gillette games, we are only able to get an upper
bound on the degree of the form $m^{O(N^2)}$ and consider getting
improved bounds for this case an interesting problem (we have no lower
bounds better than for the case of Shapley games). As explained above,
replacing the big-Os with explicit constants requires ``big-O-less''
versions of the quantifier elimination theorems and sampling theorems
of semi-algebraic geometry. We acknowledge that it is a
straightforward but also probably quite work-intensive task to
understand exactly which constants are implied by existing
proofs. Clearly, we would be interested in such results, and are
encouraged by recent work of the real algebraic geometry community
\cite{BasuBall} essentially providing a big-O-less version of the very
related Theorem 13.15 of Basu, Pollack and Roy.  We do hypothesize
that the constants will be much worse that the constant of our
big-O-less version of Corollary 13.18 of Basu, Pollack and Roy and
that merely stating some constants would stimulate work improving
them.


As a final
byproduct to our techniques, we give a new upper bound on the complexity of the strategy
iteration algorithm for concurrent reachability games
\cite{ChatQest} that matches the known lower bound \cite{CRGIteration}. 
We show:
{\em The strategy improvement algorithm of Chatterjee, de Alfaro and
Henzinger \cite{ChatQest} computes an $\epsilon$-optimal strategy
in a concurrent reachability game with $N$ actions, $m$ actions for each
player in each position after at most $(1/\epsilon)^{m^{O(N)}}$ 
iterations.}
Prior to this paper only a doubly exponential upper bound on the complexity
of strategy iteration was known, even for the case of a constant
number of positions \cite{CRGIteration}. The proof uses a known connection
between the {\em patience} of concurrent reachability games and the
convergence rate of strategy iteration \cite{CRGIteration} 
combined with a new bound on the patience proved using a somewhat more clever use of semi-algebraic geometry than in the work leading to the previous bound \cite{LICS:HKM09}.

\subsubsection*{Structure of the paper}
Section~\ref{sec-pre} contains background material and notation.
Section~\ref{sec-alg} contains a description of our
algorithms. Section~\ref{sec-upp} contains the upper bounds on degree
of values and lower bounds on coefficient sizes of defining
polynomials and resulting separation bounds of values needed for the
algorithm, for the case of Shapley, Everett and Gillete games.  The
proof of the exact bounds, big-O-less version, of the algebraic degree
and the separation bounds of the isolated real solutions of a
polynomial system is presented in Section~\ref{sec-isolated}.
Section~\ref{sec-low} presents the lower bound construction on the
Shapley games and the algebraic tools needed for this. Finally,
Section~\ref{sec-si} contains the consequences of our results for the
strategy improvement algorithm for concurrent reachability are
explained.

\section{Preliminaries}
\label{sec-pre}
\subsubsection*{(Parameterized) Matrix Games}

A matrix game is given by a real $m \times n$ matrix $A$ of 
payoffs $a_{ij}$. 
When Player I plays action $i \in \{1,2,\ldots,m\}$ and Player II
simultaneously plays action $j \in \{1,2,\ldots,n\}$, Player I receives a payoff $a_{ij}$ from Player
II. A strategy of a player is a probability
distribution over the player's actions, i.e. a stochastic vector. 
Given strategies $x$ and $y$ for the two players,
the expected payoff to player I is $x^\transpose Ay$. We denote by
$\val(A)$ the maximin value of the game.
As is well-known the value as well as an
optimal mixed strategy for Player I can be found by the following
linear program, in variables $x_1,\dots,x_m$ and $v$. By $f_n$ we
denote the vector of dimension $n$ with all entries being 1.
\begin{equation}
\label{EQ:MatrixLP}
\begin{array}{lrcl}
\max & \multicolumn{1}{l}{v}\\
\textit{s.t.} &f_nv - A^\transpose x & \leq & 0\\
&x & \geq & 0\\
&f_m^\transpose x & =& 1\\
\end{array}
\end{equation}
The following easy lemma of Shapley is useful.
\begin{lemma}[\cite{PNAS:Shapley53}, equation (2)]
\label{LEM:MatrixGameInfnormBound}
Let $A=(a_{ij})$ and $B=(b_{ij})$ be $m \times n$ matrix games. Then
\[
|\val(A) - \val(B)| \leq \max_{i,j} |a_{ij}-b_{ij}|
\]
\end{lemma}
In the following we will find it convenient to use terminology of
Bertsimas and Tsitsiklis \cite{BertsimasTsitsiklis}. We say that a set
of linear constraints are linearly independent if the corresponding
coefficient vectors are linearly independent.
\begin{definition}
  Let $P$ be a polyhedron in $\RR^n$ defined by linear equality and
  inequality constraints and let $x \in \RR^n$.
\begin{enumerate}
\item $x$ is a \emph{basic solution} if all equality constraints of $P$ are
  satisfied by $x$, and there are $n$ linearly independent constraints
  of $P$ that are satisfied with equality by $x$.
\item $x$ is a \emph{basic feasible solution} (bfs) if $x$ is a basic
  solution and furthermore satisfies all the constraints of $P$.
\end{enumerate}
\end{definition}
The polyhedron defined by LP~(\ref{EQ:MatrixLP}) is given by 1
equality constraint and $n+m$ inequality constraints, in $m+1$
variables. Since the polyhedron is bounded, the LP obtains its optimum
value at a bfs. To each bfs, $(x,v)$, we may thus associate a set of
$m+1$ linearly independent constraints such that turning all these
constraints into linear equations yields a linear system where $(x,v)$
is the unique solution. Furthermore we may express this solution using
Cramer's rule.  We order the variables as $x_1,\dots,x_m,v$, and we
also order the constraints so that the equality constraint is the last
one. Let $B$ be a set of $m+1$ constraints of the linear program,
including the equality constraint. We shall call such a set $B$ a
\emph{potential basis set}.  Define $M^A_B$ to be the $(m+1) \times
(m+1)$ matrix consisting of the coefficients of the constraints in
$B$. The linear system described above can thus be succinctly stated
as follows:
\begin{equation*}
M^A_B \begin{bmatrix}x\\v\end{bmatrix} = e_{m+1} \enspace .
\end{equation*}

We summarize the discussion above by the following lemma.
\begin{lemma}Let $v \in \RR$ and $x \in \RR^m$ be given.
\label{LEM:LPbfs}
\begin{enumerate}
\item The pair $(x,v)^\transpose$ is a basic solution of (\ref{EQ:MatrixLP})
  if and only if there is a potential basis set $B$ such
  that $\det(M^A_B) \neq 0$ and $(x,v)^\transpose = (M^A_B)^{-1}e_{m+1}$.
\item The pair $(x,v)^\transpose$ is a bfs of (\ref{EQ:MatrixLP}) if and only if
  there is a potential basis set $B$ such that $\det(M^A_B) \neq 0$,
  $(x,v)^\transpose = (M^A_B)^{-1}e_{m+1}$, $x \geq 0$ and $f_nv - A^\transpose x \leq
  0$.
\end{enumerate}
\end{lemma}

By Cramer's rule we find that $x_i =\det((M^A_B)_i)/\det(M^A_B)$ and
$v = \det((M^A_B)_{m+1})/\det(M^A_B)$. Here $(M^A_B)_i$ is the matrix
obtained from $M^A_B$ by replacing column $i$ with $e_{m+1}$.

We shall be interested in \emph{parameterized} matrix games. Let $A$
be a mapping from $\RR^N$ to $m \times n$ matrix games. Given a
potential basis set $B$ we will be interested in describing the sets
of parameters for which $B$ gives rise to a bfs as well as an optimal
bfs for LP~(\ref{EQ:MatrixLP}). We let $F^A_B$ denote the set of $w
\in \RR^N$ such that $B$ defines a bfs for the matrix game $A(w)$, and
we let $O^A_B$ denote the set of $w \in \RR^N$ such that $B$ defines
an optimal bfs for the matrix game $A(w)$. Let $\overline{B_1}
\subseteq \{1,\dots,n\}$ be the set of indices of the first $n$
constraints that are not in $B$. Similarly, let $\overline{B_2}
\subseteq \{1,\dots,m\}$ be the indices of the next $m$
constraints that are not in $B$. We may describe the set $F^A_B$ as a
union $F^{A+}_B \cup F^{A-}_B$. Here $F^{A+}_B$ is defined to be the
set of parameters $w$ that satisfy the following $m+1$ inequalities:
\[
\begin{split}
\det(M^{A(w)}_B) &> 0 \enspace ,\\
\det((M^{A(w)}_B)_{m+1}) - \sum_{i=1}^m a_{ij} \det(((M^{A(w)}_B)_{i})) &\leq 0 \text{\ for}\ j \in \overline{B_1},\\
\det((M^{A(w)}_B)_i) &\geq 0 \text{\ for}\ i \in \overline{B_2}.
\end{split}
\]
The set $F^{A-}_B$ is defined analogously, by reversing all
inequalities above. With these in place we can describe $O^A_B$ as the
sets of parameters $w \in F^A_B$ for which
\[
\det((M^{A(w)}_B)_{m+1}) = \val(A(w))\det(M^{A(w)}_B) \enspace .
\]

\subsubsection*{Shapley and Everett games}

We will define stochastic games in a general form, following Everett
\cite{AMS:Everett57}, to capture both Shapley games as well as Everett
games (but not Gillette games) as direct specializations.  Everett in
fact defined his games abstractly in terms of ``game elements''. We
shall restrict ourselves to game elements that are given by matrix
games (cf.\ \cite{JAP:Orkin72}). Because of this, our precise
notation will differ slightly from the one of Everett.

For that purpose a stochastic game $\Gamma$ is specified as
follows. We let $N$ denote the number of \emph{positions}, numbered
$\{1,\dots,N\}$. In every position $k$, the two players have
$m_k$ and $n_k$ \emph{actions} available, numbered $\{1,\dots,m_k\}$
and $\{1,\dots,n_k\}$. If at position $k$ Player I chooses action $i$
and Player II simultaneously chooses action $j$, Player I receives reward $a^k_{ij}$
from player II. After this, with probability $s^k_{ij} \geq 0$ the game stops,
in which case Player I receives an additional reward $b^k_{ij}$ from player
II. With probability $p^{kl}_{ij} \geq 0$, play continues at position $l$.
We demand $s^k_{ij}+\sum_{l=1}^N p^{kl}_{ij} = 1$ for
all positions $k$ and all pairs of actions $(i,j)$.
A
\emph{strategy} of a player is an assignment of a probability
distribution on the actions of each position, for each possible
history of the play, 
a history being the sequence of positions visited so far as well as the
sequences of actions played by both players in those rounds. 
A strategy is called \emph{stationary} if it only depends on
the current position.

Given
a pair of strategies $x$ and $y$ as well as a starting position $k$,
let $r_i$ be the random variable denoting the reward given
to Player I during round $i$ (if play has ended we define this as
0). We define the expected total payoff by 
$\tau^k(x,y) = \lim_{n \rightarrow \infty} E\left[ \sum_{i=1}^n r_i \right] \enspace ,
$
where the expectation is taken over actions of the players according
to their strategies $x$ and $y$, as well as the probabilistic choices of
the game (In the special cases of Shapley and Everett games the
limit always exist).
We define the \emph{lower value}, $\underline{v}^k$, and \emph{upper
  value}, $\overline{v}^k$, of the game $\Gamma$, starting in position
$k$ by 
$
\underline{v}^k = \sup_x \inf_y \tau^k(x,y),
$
and
$
\overline{v}^k = \inf_y \sup_x \tau^k(x,y).
$
In case that $\underline{v}^k=\overline{v}^k$ we define this as the
\emph{value} $v^k$ of the game, starting at position $k$.
Assuming $\Gamma$ has a value, starting at position $k$, we say that a
strategy $x$ is \emph{optimal} for Player I, starting at position $k$
if
$
\inf_y \tau^k(x,y) = v^k ,
$
and for a given $\epsilon>0$, we say the strategy $x$ is $\epsilon$-optimal starting at position $k$, if
$
\inf_y \tau^k(x,y) \geq v^k - \epsilon \enspace .
$
We define the notions of optimal and $\epsilon$-optimal analogously
for Player II. 
%

A Shapley game \cite{PNAS:Shapley53} is a special case of the above
defined stochastic games, where $s^k_{ij} > 0$ and $b^k_{ij}=0$ for
all positions $k$ and all pairs of actions $(i,j)$.
%
Given \emph{valuations} $v_1,\dots,v_N$ for the positions and a given
position $k$ we define $A^k(v)$ to be the $m_k \times n_k$ matrix
game where entry $(i,j)$ is
$
a^k_{ij} + \sum_{l=1}^N p^{kl}_{ij}v_l. $
The \emph{value iteration} operator $T : \RR^N \rightarrow \RR^N$ is
defined by $T(v) = \left(\val(A^1(v)),\dots,\val(A^N(v))\right)$.
The following theorem of Shapley characterizes the value and optimal strategies of a Shapley game.
\begin{theorem}[Shapley]
\label{THM:ShapleyValue}
  The value iteration operator $T$ is a
  contraction mapping with respect to supremum norm.
  In particular, $T$ has a unique fixed point, and this is the
  value vector of the stochastic game $\Gamma$. 
  Let $x^*$ and $y^*$ be
  the stationary strategies for Player I and player II where in position $k$ an
  optimal strategy in the matrix game $A^k(v^*)$ is played. Then $x^*$
  and $y^*$ are optimal strategies for player I and player II,
  respectively, for play starting in any position.
\end{theorem}

An Everett game \cite{AMS:Everett57} is a special case of the above
defined stochastic games, where $a^k_{ij} = 0$ for all $k,i,j$. In contrast to
Shapley games, we may have that $s^k_{ij}=0$ for some $k,i,j$.
  Everett points out that his games generalize the class of
  Shapley games. Indeed, we can convert Shapley game $\Gamma$ to 
  Everett game $\Gamma'$ by letting $b^k_{ij} =
  a^k_{ij}/s^k_{ij}$, recalling that $s^k_{ij}>0$. 

Given \emph{valuations} $v_1,\dots,v_N$ for the positions and a given
position $k$ we define $A^k(v)$ to be the $m_k \times n_k$ matrix
game where entry $(i,j)$ is
$
s^k_{ij}b^k_{ij} + \sum_{l=1}^N p^{kl}_{ij}v_l.$
The \emph{value mapping} operator $M : \RR^N \rightarrow \RR^N$ is
then defined by
$
M(v) = \left(\val(A^1(v)),\dots,\val(A^N(v))\right)$.
Define relations $\dotgeq$ and $\dotleq$ on $\RR^N$ as follows:
\begin{gather*}
u \dotgeq v \quad \text{if and only if} \quad \begin{cases}u_i > v_i & \text{if}\ v_i>0\\ u_i\geq v_i & \text{if}\ v_i \leq 0\end{cases} \enspace , \quad \text{for all}\ i \enspace .\\
u \dotleq v \quad \text{if and only if} \quad \begin{cases}u_i < v_i & \text{if}\ v_i<0\\ u_i\leq v_i & \text{if}\ v_i \geq 0\end{cases} \enspace , \quad \text{for all}\ i \enspace .\\
\end{gather*}
Next, we define the regions $C_1(\Gamma)$ and $C_2(\Gamma)$ as follows:
\begin{gather*}
C_1(\Gamma) = \{v \in \RR^N \mid M(v) \dotgeq v\},\\
C_2(\Gamma) = \{v \in \RR^N \mid M(v) \dotleq v\}.
\end{gather*}
A \emph{critical vector} of the game is a vector $v$ such that $v \in
\overline{C_1(\Gamma)} \cap \overline{C_2(\Gamma)}$. That is, for every
$\epsilon>0$ there exists vectors $v_1 \in C_1(\Gamma)$ and $v_2 \in C_2(\Gamma)$ such that
$\norm{v-v_1}_2 \leq \epsilon$ and $\norm{v-v_2}_2 \leq \epsilon$.

The following theorem of Everett characterizes the value of an Everett game and exhibits near-optimal strategies.
\begin{theorem}[Everett]
\label{THM:CriticalVectorExistence}
  There exists a unique critical vector $v$ for the value mapping $M$,
  and this is the value vector of $\Gamma$. Furthermore,
  $v$ is a fixed point of the value mapping, and if $v_1 \in C_1(\Gamma)$ and
  $v_2 \in C_2(\Gamma)$ then $v_1 \leq v \leq v_2$.
  Let $v_1 \in C_1(\Gamma)$. Let $x$ be the stationary strategy 
  for player I, where in position $k$ an optimal strategy in the matrix game $A^k(v_1)$ is played. Then for any $k$, starting play in position $k$, the
  strategy $x$ guarantees expected payoff at least $v_{1,k}$
  for player I. The analogous statement holds for $v_2 \in C_2(\Gamma)$ and Player II.
\end{theorem}

\subsubsection*{Gillette Games}
\label{def-gil}
While the payoffs in Gillette's model of stochastic games cannot be captured as a special case of the general formalism above, the general setup 
is the same, i.e., the parameters $N, m_k, n_k, a^k_{ij}, p^{kl}_{ij}$ is as above and the game is played as in the case of Shapley games and Everett games. 
In Gillette's model, we have $b^k_{ij} = 0$ and $s^k_{ij}=0$ for all $k,i,j$. 
The payoff associated with an 
infinite play of a Gillette game is by definition 
$\liminf_{T \rightarrow \infty} (\sum_{t=1}^T r_i)/T$ where $r_t$ is the reward collected at round $t$. 
Upper and lower values are defined analogously to the case of Everett and Shapley games, but with the expectation of the payoff defined in this way replacing $\tau^k(x,y)$. Again, the value of position $k$ is said to exist if its upper and lower value coincide. 
An Everett game can be seen as a special case of a Gillette game by replacing each termination outcome with final reward $b$ with an absorbing position in which the reward $b$ keeps recurring.
The central theorem about Gillette games is
the theorem of Mertens and Neyman \cite{MN}, showing that all such games have a value. The proof also yields the following connection to Shapley games that is used by our algorithm: For a given Gillette game $\Gamma$, let $\Gamma_\lambda$ be the Shapley game with all stop probabilities
$s^k_{ij}$ being $\lambda$ and each transition probability being the corresponding transition probability of $\Gamma$ multiplied by $1-\lambda$. Let $v^k$ be the value of position $k$ in $\Gamma$ and let $v_\lambda^k$ be the value of position $k$ in $\Gamma_\lambda$. Then, the following holds. 
\begin{theorem}[Mertens and Neyman]\label{THM:MertensNeyman}
\[ v^k = \lim_{\lambda \rightarrow 0+} \lambda v^k_\lambda \] 
\end{theorem}

\subsubsection*{Real Algebraic Numbers}

Let $p(x) \in \ZZ[x]$ be a nonzero polynomial with integer
coefficients of degree $d$. Write $p(x) = \sum_{i=1}^d a_i x^i$, with
$a_d \neq 0$. The \emph{content} $\cont(p)$ of $p$ is defined by
$\cont(p) = \gcd(a_0,\dots,a_d)$, and we say that $p$ is
\emph{primitive} if $\cont(p)=1$. We can view the coefficients of $p$
as a vector $a \in \RR^{d+1}$. We then define the \emph{length}
$\abs{p}$ of $p$ by $\abs{p} = \norm{a}_2$ as well as the
\emph{height} $\height{p}$ of $p$ by $\height{p} = \norm{a}_\infty$.

An algebraic number $\alpha \in \CC$ is a root of a polynomial in
$\QQ[x]$. The \emph{minimal polynomial} of $\alpha$ is the unique
monic polynomial in $q \in \QQ[x]$ of least degree with $q(\alpha)=0$.
Given an algebraic number $\alpha$ with minimal polynomial $q$, there
is a minimal integer $k \geq 1$ such that $p=kq \in \ZZ[x]$. In other
words $p$ is the unique polynomial in $\ZZ[x]$ of least degree with
$p(\alpha)=0$, $\cont(p)=1$ and positive leading coefficient. We
extend the definitions of degree and height to $\alpha$ from $p$.  The
\emph{degree} $\deg(\alpha)$ of $\alpha$ is defined by
$\deg(\alpha)=\deg(p)$ and \emph{height} $\height{\alpha}$ of
$\alpha$ is defined by $\height{\alpha} = \height{p}$.

\begin{theorem}[Kannan, Lenstra and Lov\'{a}sz]\label{KLL}
  There is an
  algorithm that computes the minimal polynomial of a given algebraic number
$\alpha$ of degree $n_0$ when
  given as input $d$ and $H$ such that $\deg(\alpha) \leq d$ and $\height{\alpha} \leq H$ and $\overline{\alpha}$ such that
  $\abs{\alpha-\overline{\alpha}} \leq 2^{-s}/(12d)$, where
\[
s = \lceil d^2/2 + (3d+4)\log_2(d+1) + 2d\log_2(H) \rceil \enspace .
\]
The algorithm runs in time polynomial in $n_0, d$ and $\log H$.
\end{theorem}

\section{Algorithms}
\label{sec-alg}

In this section we describe our algorithms for solving Shapley, Everett and Gillette games. The algorithms for Shapley and Everett games proceed along the same lines, using the fact that Shapley games can be seen as a special case of Everett games explained above. The algorithm for Gillette games is a reduction to the case of Shapley games using Theorem~\ref{THM:MertensNeyman}.
We proceed by first constructing the algorithms for Everett and Shapley games and explain the algorithm for Gillette games at the end of this section. 
\subsubsection*{Reduced games}

\label{sec:reduced-games}

Let $\Gamma$ be an Everett game with $N+1$ positions. Denote by
$V(\Gamma)$ the critical vector of $\Gamma$. Given a valuation $v$ for
position $N+1$ we consider the \emph{reduced game} $\Gamma^r(v)$ with
$N$ positions, obtained from $\Gamma$ in such a way that whenever the
game would move to position $N+1$, instead the game would stop and
player 1 would receive a payoff $v$.

Denote by $V^r(v)$ the critical vector of the game $\Gamma^r(v)$. We have
the following basic lemma shown by Everett.
\begin{lemma}
\label{LEM:ReducedGameContinuity}
For every $\delta>0$, for all $v$ and for all positions $k$:
$(V^r(v))_k - \delta \leq (V^r(v-\delta))_k \leq (V^r(v))_k $
$\leq$
$(V^r(v+\delta))_k \leq (V^r(v))_k + \delta.$
In particular, $V^r(v)$ is a continuous monotone
function of $v$ in all components.
The first and last inequalities are strict inequalities,
unless $(V^r(v))_k=v$. 
\end{lemma}

Let $\widetilde{V}(v)$ denote the value $\val(A^{N+1}(V^r(v),v))$
of the parameterized game for position $N+1$, where the first $N$
positions are given valuations according to $V^r(v)$ and position
$N+1$ is given valuation $v$. 

\begin{lemma}
\label{LEM:BisectionEquivalences}
  Denote by $v^*$ component $N+1$ of $V(\Gamma)$. Then the following
  equivalences hold.
\begin{enumerate}
\item   Suppose $v^*>0$ and $v \geq 0$. Then, $\widetilde{V}(v) > v \Leftrightarrow v < v^*$.
\item Suppose $v^*<0$ and $v \leq 0$. Then, $\widetilde{V}(v) <
  v \Leftrightarrow v^* < v$.
\end{enumerate}
\end{lemma}
\begin{proof}
  We shall prove only the first equivalence. The proof of the second
  equivalence is analogous.  Assume first that $\widetilde{V}(v) >
  v$. Since $\widetilde{V}$ is continuous we can find $z \in
  C_1(\Gamma^r(v))$ such that $\val(A^{N+1}(z,v))>v$ as well. This
  implies that $(z,v) \in C_1(\Gamma)$ and by definition of
  $C_1(\Gamma)$ we obtain that $v \leq v^*$. By
  Theorem~\ref{THM:CriticalVectorExistence}, $\widetilde{V}(v^*) =
  \val(A^{N+1}(V^r(v^*),v^*))$ $=$ $\val(A^{N+1}(V(\Gamma)))=v^*$.
  Since $\widetilde{V}(v) > v$ we have $v < v^*$.

  The other part of the equivalence was shown by Everett as a part of
  his proof of Theorem~\ref{THM:CriticalVectorExistence}. We present
  the argument for completeness. Everett in fact shows that $v^*$ is
  the fixpoint of $\widetilde{V}$ of minimum absolute value. That is,
  $\widetilde{V}(v^*) = v^*$ and whenever $\widetilde{V}(v) = v$ we
  have $\abs{v} \geq \abs{v^*}$. Now assume that $v < v^*$, and let
  $\delta = v^* - v$. From Lemma~\ref{LEM:ReducedGameContinuity} we
  have $\widetilde{V}(v) = \widetilde{V}(v^*-\delta) \geq
  \widetilde{V}(v^*) - \delta = v^* - \delta = v$. Since $v \geq 0$,
  from minimality of $\abs{v^*}$ we have the strict
  inequality $\widetilde{V}(v) > v$.
\end{proof}

\subsubsection*{Recursive bisection algorithm}

\begin{algorithm2e}
\caption{Bisect($\Gamma, k$) \label{alg:rba}}
\KwIn{Game $\Gamma$ with $N+1$ positions, all payoffs between -1 and
  1, accuracy parameter $k \geq 2$.}
\KwOut{$v$ such that $\abs{v-v^*}\leq 2^{-k}$.} 
\eIf{$\widetilde{V}(0)=0$}{
  \Return 0 \;
}{
  $v_l \leftarrow 0$ \;
  $v_r \leftarrow \sgn(\widetilde{V}(0))  $ \nllabel{alg:rba-sgn} \;
  \For{$i \leftarrow 1$ \KwTo $k-1$}{
    $v \leftarrow (v_l+v_r)/2$ \;
    \eIf{$\abs{\widetilde{V}(v)}>\abs{v}$ \nllabel{alg:rba-if}}{
      $v_l \leftarrow v$ \;
    }{ 
      $v_r \leftarrow v$ \;
    }
  }
  \Return $(v_l + v_r)/2$ \;
}
\end{algorithm2e} 

Based on Lemma~\ref{LEM:BisectionEquivalences} we may construct an
idealized bisection algorithm Bisect (Algorithm~\ref{alg:rba}) 
for approximating the last component of the
critical vector, unrealistically assuming we can compute the critical vector of a reduced game exactly.
For convenience and without loss of generality, we will
assume  throughout that
the payoffs in the game $\Gamma$ we consider have been
normalized to belong to the interval $[-1,1]$. 
The correctness of the algorithm follows directly from 
Lemma~\ref{LEM:BisectionEquivalences}.
Given that we have obtained a
sufficiently good approximation for the last component of the critical
vector we may reconstruct the exact value using Theorem~\ref{KLL}.
What ``sufficiently good'' means depends on the algebraic degree and size of coefficients of the defining polynomial of the algebraic number to be given as output, so we shall need bounds on these quantities for the game at hand. 

\begin{algorithm2e}
\caption{ApproxBisect($\Gamma,k$) \label{alg:appx-bisect}}
\KwIn{Game $\Gamma$ with $N+1$ positions, $m$ actions per player in each position, all payoffs rationals between -1 and
  1 and of bitsize $L$, accuracy parameter $k \geq 2$.}
\KwOut{$v$ such that $\abs{v-v^*}<2^{-k}$.}
  $\epsilon \leftarrow \sep(N,m,L,0)/5$ \;
  $v \leftarrow \val(A^{N+1}( [\text{ApproxVal}(V^r(0),\lceil - \log \epsilon \rceil)]_{\lceil -\log \epsilon  \rceil},0))$ \;
  \eIf{$\abs{v} \leq 2\epsilon$}{ \Return 0 }{
  $v_l \leftarrow 0$ \;
  $v_r \leftarrow \sgn(v)$ \;
   \For{$i \leftarrow 1$ \KwTo $k-1$}{  
    $v \leftarrow (v_l+v_r)/2$ \;
    $\epsilon \leftarrow \sep(N,m, \max(L,i),i)$/5 \;
    $v' \leftarrow \val(A^{N+1}([\text{\rm ApproxVal}(V^r(v),
      \lceil - \log \epsilon \rceil)]_{\lceil -\log \epsilon  \rceil},v))$\; 
    \eIf{$\abs{v'}>\abs{v}$}
    {
      $v_l \leftarrow v$ \;
    }{ 
      $v_r \leftarrow v$ \;
    }
  }
  \Return $(v_l + v_r)/2$ \;
}
\end{algorithm2e} 

\begin{algorithm2e}
\caption{ApproxVal($\Gamma,k$) \label{alg:appx-rba}}
\KwIn{Game $\Gamma$ with $N$ positions, payoffs between -1 and 1, accuracy parameter $k \geq 2$.}
\KwOut{Value vector $v$ such that $\abs{v_i-v_i^*}< 2^ {-k}$ for all
  positions $i$.}
\eIf{$N=0$}{\Return The empty vector } 
{\For{$i \leftarrow 1$ \KwTo $N$}
{$v_i=\text{ApproxBisect}(\Gamma,k)$,
where position $i$ is swapped with position $N$}
Return $v$}
\end{algorithm2e}

To get an algorithm implementable as a Turing machine we will have to
compute with approximations
throughout the algorithm but do so in a way that simulates
Algorithm~\ref{alg:rba} exactly, i.e., so that the same branches
are followed in the if-statements of the algorithm. 
For this, we need separation bounds
for values of stochastic games. Fortunately, these follow from the bounds on degree and coefficient size needed anyway to apply Theorem~\ref{KLL}.
  Consider a class $\mathcal{C}$ of Everett games (In fact
  $\mathcal{C}$ will be either all Everett games or the subset consisting of Shapley games).  Let
  $\sep(N,m,L,j)$ denote a positive real number so that if $v$ is
the value of game $\Gamma \in \mathcal{C}$ with $N$ positions,
  $m$ actions to each player in every position, and every rational
  occurring in the description in the game having bitsize at most
  $L$, and $v$ is not an integer multiple of $2^{-j}$, then $v$ differs
by at least $\sep(N,m,L,j)$ from every integer multiple of $2^{-j}$. 
Also, we let $[v]_k$ denote the function that rounds all entries
in the vector $v$ to the nearest integer multiple of $2^{-k}$. 
Our modified algorithm ApproxBisect (for approximate Bisect) is given as Algorithm~\ref{alg:appx-bisect}. The procedure ApproxVal invoked in the code simply 
computes approximations to the values of all positions in a game
using ApproxBisect.

The correctness of ApproxBisect follows from the correctness of Bisect by
observing that the former emulates the latter, in the sense that the
same branches are followed in the if-statements.  For the latter fact,
Lemma~\ref{LEM:MatrixGameInfnormBound} and
Lemma~\ref{LEM:BisectionEquivalences} are used.

The complexity of the algorithm is estimated by the
inequalities
\[
T_\text{ApproxVal}(N,m,L,k) \leq N T_\text{ApproxBisect}(N,m,L,k),
\]
and 
\[
\begin{split}
T_\text{ApproxBisect}(N,m,L,k) & \leq  \lceil- \log \epsilon\rceil (T_\text{LP}(m+1,
\lceil - \log \epsilon \rceil) \\ &  + 
T_\text{ApproxVal}(N-1,m,\max\{L,k\},\lceil - \log \epsilon \rceil),
\end{split}
\] 
where $\epsilon=\sep(N-1,m, \max\{L,k\},k)/5$, and
$T_\text{LP}(m+1,k)$ is a bound on the complexity of computing the
value of a $m \times m$ matrix game with entries of bitsize $k$.

Plugging in the separation bound for Shapley games of
Proposition~\ref{PROP:sepShap}, we get a concrete algorithm without
unspecified constants.  Also, to get an algorithm that outputs the
exact algebraic answer in isolating interval encoding we need to call
the algorithm with parameter $k$ appropriately chosen to match the
quantities stated in Theorem~\ref{KLL}, taking into account the degree
and coefficient bounds given in Proposition~\ref{PROP:sepShap}.
Finally, plugging in a polynomial bound for $T_\text{LP}$, the above
recurrences is now seen to yield a polynomial time bound for constant
$N$. However, the exponent in this polynomial bound is $O(N)^{N^2}$,
i.e., the complexity is doubly exponential in $N$.  We emphasize that
the fact that the exact value is reconstructed in the end only
negligibly changes the complexity of the algorithm compared to letting
the algorithm return a crude approximation. Indeed, an approximation
algorithm following our approach would have to compute with a
precision in its recursive calls similar to the precision necessary
for reconstruction. Only for games with only one position (and hence
no recursive calls) would an approximation version of ApproxBisect be
faster.

For the case of Everett games, the degree, coefficient and separation bounds of
Proposition~\ref{PROP:sepEve}
similarly yields the existence of a polynomial time algorithm for the
case of constant $N$, with an exponent of $N^{O(N^2)}$.
%

\subsubsection*{Computing strategies}

We now consider the task of computing $\epsilon$-optimal strategies to
complement our algorithm for computing values. For Shapley games the
situation is simple. By Theorem~\ref{THM:ShapleyValue}, once we
have obtained the value $v^*$ of the game, we can obtain
\emph{exactly} optimal stationary strategies $x^*$ and $y^*$ by finding
optimal strategies in the matrix games $A^k(v^*)$. Also, if we only
have an approximation $\tilde{v}$ to $v^*$, such that
$\norm{v^*-\tilde{v}}_\infty \leq \epsilon$, consider the stationary
strategies $\tilde{x}^*$ and $\tilde{y}^*$ given by optimal strategies
in the matrix games $A^k(\tilde{v})$. In every round of play, these
strategies may obtain $\epsilon$ less than the optimal strategies. But
this deficit is \emph{discounted} in every round by a factor $1-\lambda$
where $\lambda=\min(s_{ij}^k)>0$ is the minimum stop probability. Hence
$\tilde{x}$ and $\tilde{y}$ are in fact $(\epsilon/\lambda)$-optimal strategies.

For Everett games the situation is more complicated, since the actual
values $v^*$ may in fact give absolutely no information about
$\epsilon$-optimal strategies. We shall instead follow the approach of
Everett and show how to find points $v_1 \in C_1$ and $v_2 \in C_2$
that are $\epsilon$-close to $v^*$. Then, using Theorem~\ref{THM:CriticalVectorExistence} we can compute
$\epsilon$-optimal strategies by finding optimal strategies in the
matrix games $A^k(v_1)$ and $A^k(v_2)$, respectively.

Let $\Gamma$ be an Everett game with $N+1$ positions.  We first
describe how to exactly compute $v_1 \in C_1$, given the ability to
exactly compute the values; the case of $v_2 \in C_2$ is analogous.
Let $v^*$ be the critical vector of $\Gamma$. In case that $v^*_i \leq
0$ for all $i$, then by definition of $C_1$ we have $v^* \in C_1$.
Otherwise at least one entry of $v^*$ is positive, so assume
$v^*_{N+1} > 0$. As in Section~\ref{sec:reduced-games} we consider the
reduced game $\Gamma^r(v)$, taking payoff $v$ for position $N+1$. By
Lemma~\ref{LEM:BisectionEquivalences}, whenever $0 \leq v < v^*_{N+1}$
we have $\widetilde{V}(v) > v$. Suppose in fact that we pick $v$ so
that $v^*_{N+1} - v \leq \epsilon/2$.  Now let $\delta =
\widetilde{V}(v) - v$. Recall $\widetilde{V}(v) =
\val(A^{N+1}(V^r(v),v))$. Now recursively compute $z \in
C_1(\Gamma^r(v))$ such that $\norm{V^r(v)-z}_\infty \leq
\min(\delta/2,\epsilon)$. Then by
Lemma~\ref{LEM:MatrixGameInfnormBound} we have that
$\abs{\val(A^{N+1}(V^r(v),v)) - \val(A^{N+1}(z,v))} \leq \delta/2$,
which means $\val(A^{N+1}(z,v)) > v$. This means that $v_1=(z,v) \in
C_1$, and by our choices we have $\norm{v_1-v^*}_\infty \leq
\epsilon$, as desired.  We now have an exact representation of an
algebraic vector $v_1$ in $C_1$, $\epsilon$-approximating the critical
vector. The size of the representation in isolating interval
representation is polynomial in the bitsize of $\Gamma$ (for constant
$N$). From this we may compute the optimal strategies of $A^{k}(v_1)$
which also form an $\epsilon$-optimal strategy of $\Gamma$. The
polynomial size bound on $v_1$ implies that all non-zero entries in
this strategy have magnitude at least $2^{-l}$ where $l$ is
polynomially bounded in the bitsize of $\Gamma$.  We now show how to
get a rational valued $2\epsilon$-optimal strategy in polynomial
time. For this, we apply a rounding scheme described in Lemmas 14 and
15 of Hansen, Kouck\'{y} and Miltersen \cite{LICS:HKM09}. For each
position, we now round all probabilities, except the largest, {\em
  upwards} to $L$ significant digits where $L$ is a somewhat larger
polynomial bound than $l$, while the largest probability at each
position is rounded downwards to $L$ significant digits.  Using Lemma
14 (see also the proof of Lemma 15) of Hansen, Kouck\'{y} and
Miltersen \cite{LICS:HKM09}, we can set $L$ so that the resulting
strategy is $2\epsilon$-optimal in $\Gamma$.  This concludes the
description of the procedure.

\subsubsection*{The case of Gillette games}
To compute the value of a given Gillette game, we proceed as follows.
Based on Theorem~\ref{THM:MertensNeyman} we can similarly to the case
of Shapley games and Everett games give degree, coefficient, and
separation bounds for the values of the given game. These are given in
Proposition~\ref{PROP:sepGil}. Furthermore, and also based on
Theorem~\ref{THM:MertensNeyman}, we can for a given $\epsilon$ give an
explicit upper bound on the value of $\lambda$ necessary for
$v^k_\lambda$ to approximate $v^k$ within $\epsilon$. This expression
for such $\lambda$, given in
Proposition~\ref{PROP:Gillette-discountbound}, is of the form
$\lambda_\epsilon = \epsilon^{\tau m^{O(N^2)}}$. Our algorithm
proceeds simply by setting $\epsilon$ so small that an
$\epsilon$-approximation to the value allows an exact reconstruction
of the value using Theorem~\ref{KLL}. Such $\epsilon$ can be computed
as we have derived degree and coefficient bounds for the value of the
Gillette game at hand. We then run our previously constructed
algorithm on the Shapley game $\Gamma_\lambda$, where $\lambda =
\lambda_\epsilon$.

 
\section{Degree and separation bounds for Stochastic Games}
\label{sec-upp}

\subsection{Shapley Games}
Our bounds on degree, coefficient size, and separation for Shapley
games are obtained by a reduction to the same question about isolated
solutions of polynomial systems. The latter is treated in Section
\ref{sec-isolated}. In this section we present the reduction as well
as stating the consequences obtained from this and Theorem
\ref{th:isol-real-root-bd-full} of Section \ref{sec-isolated}.
To analyse our reduction we also need the following simple fact.
\begin{proposition}[\cite{BasuPollackRoy2006}, Proposition 8.12]
\label{PROP:DetPolyBound}
  Let $M$ be an $m \times m$ matrix, whose entries are integer
  polynomials in variables $x_1,\dots,x_n$ of degree at most $d$ and
  coefficients of bitsize at most $\tau$. Then $\det(M)$, as a
  polynomial in variables $x_1,\dots,x_n$ is of degree at most $dm$
  and has coefficients of bitsize at most $(\tau + \bit(m))m + n\bit(md+1)$,
  where $bit(z)=\lceil \lg z \rceil$.
\end{proposition}

Also, denote by $\ball(v,\epsilon)$ the ball around $v \in \RR^N$ of
radius $\epsilon>0$, $\{v' \in \RR^N \mid \norm{v-v'}_2 \leq
\epsilon\}$. We are now in position to present the reduction.
\begin{theorem}
  Let $\Gamma$ be a Shapley game, with $N$ positions. Assume that in
  position $k$, the two players have $m_k$ and $n_k$ actions
  available. Assume further that all payoffs and probabilities in
  $\Gamma$ are rational numbers with numerators and denominators of
  bitsize at most $\tau$.

  Then there is a system $\mathcal{S}$ of polynomials in variables
  $v_1,$ $\ldots,$ $v_N$, for which the value vector $v^*$ of $\Gamma$ 
  is an isolated
    root. Furthermore the system $\mathcal{S}$ consists of at most
  $\sum_{k=1}^N \binom{n_k+m_k}{m_k}$ polynomials, each of degree at
  most $m+2$ and having integer coefficients of bitsize at most
  $2(N+1)(m+1)^2\tau+1$, where $m =
  \max_{k=1}^N\left(\min(n_k,m_k)\right)$.
  \label{THM:ShapleyPolynomialSystem}
\end{theorem}
\begin{proof}
  Let $v^* \in \RR^n$ be the fixpoint of $T$ given by
  Theorem~\ref{THM:ShapleyValue}. For all positions $k$, and for all
  potential basis sets $B^k$ corresponding to the parameterized matrix
  game $A^k$ we consider the closures $\overline{O^{A^k}_{B^k}}$ of
  the sets $O^{A^k}_{B^k}$. Since there are finitely many positions
  and for each position finitely many potential basis sets, we may
  find $\epsilon > 0$ such that whenever $\ball(v^*,\epsilon) \cap
  \overline{O^{A^k}_{B^k}} \neq \emptyset$ we have $v^* \in
  \overline{O^{A^k}_{B^k}}$ for all positions $k$ and all potential
  basis sets $B^k$. For a given position $k$, let $\mathcal{B}^k$ be
  the set of such potential basis sets. Then, for every $B^k \in
  \mathcal{B}^k$ define the polynomial
  \begin{equation*}
    \label{EQ:FixpointPoly}
    P_{B^k}(w) = \det((M^{A^k(w)}_{B^k})_{m_k+1}) - w_k\det(M^{A^k(w)}_{B^k}) \enspace .
  \end{equation*}
  Let $\mathcal{P}$ be the system of polynomials consisting of all such
  polynomials for all positions $k$. We claim that $v^*$ is an isolated
  root of the system $\mathcal{P}$. First we show that $v^*$ is in fact
  a solution. Consider any position $k$ and any polynomial $P_{B^k} \in
  \mathcal{P}$. By construction we have $v^* \in
  \overline{O^{A^k}_{B^k}}$, and we may thus find a sequence
  $(w^i)_{i=1}^\infty$ in $O^{A^k}_{B^k}$ converging to $v^*$. Since for
  every $i$, $w^i \in O^{A^k}_{B^k}$ we have that
  $\det((M^{A^k(w^i)}_B)_{m+1}) - \val(A^k(w^i))\det(M^{A(w^i)}_B) = 0$,
  and thus by continuity of the functions $\det$, $\val$, and the
  entries of $A^k$, we obtain $\det((M^{A^k(v^*)}_B)_{m+1}) -
  \val(A^k(v^*))\det(M^{A^k(v^*)}_B)=0$. But $\val(A^k(v^*)) = v^*_k$
  and hence $P_{B^k}(v^*)=0$.

  Next we show that $v^*$ is unique. Indeed, suppose that $v' \in
  \ball(v^*,\epsilon)$ is a solution to the system $\mathcal{P}$. For
  each position $k$ pick a potential basis set $B^k$ such that $B^k$
  describes an optimal bfs for $A^k(v')$. Now since $v' \in
  \ball(v^*,\epsilon)$ as well as $v' \in O^{A^k}_{B^k}$ we
  have by definition that $B^k \in \mathcal{B}^k$ and hence $P_{B^k} \in
  \mathcal{P}$. As a consequence $v'$ must be a root of $P_{B^k}$. Now,
  since $B^k$ in particular is a basic solution we have
  $\det(M^{A^k(v')}_{B^k}) \neq 0$. Combining these two facts we obtain
  \[
  v'_k = \det((M^{A^k(v')}_{B^k})_{m_k+1})/\det(M^{A^k(v')}_{B^k}) \enspace ,
  \]
  and since $B^k$ is an optimal bfs for $A^k(v')$ we have that
  $\val(A^k(v'))_k = v'_k$. Since this holds for all $k$, we obtain that
  $v'$ is a fixpoint of $T$, and Theorem~\ref{THM:ShapleyValue} then
  gives that $v'=v^*$.

  To get the system $\mathcal{S}$ we take (smallest) integer multiples
  of the polynomials in $\mathcal{S}$ such that all polynomials have
  integer coefficients. For a given position $k$, we have
  $\binom{n_k+m_k}{m_k}$ potential basis sets, giving the bound on the
  number of polynomials. Assume now that $m_k \leq n_k$ (In case
  $m_k>n_k$ we can consider the dual of the linear program in
  Lemma~\ref{LEM:LPbfs}). Fix a potential basis set $B^k$.

  Using Proposition~\ref{PROP:DetPolyBound} the degree of $P_{B^k}(w)$
  is at most $1+(m_k+1)$. Further to bound the bitsize of the
  coefficients, note that using linearity of the determinant we may
  multiply each row of the matrices $(M^{A^k(w)}_{B^k})_{m_k+1}$ and
  $M^{A^k(w)}_{B^k}$ by the product of the denominators of all the
  coefficients of entries in the same row in the matrix
  $M^{A^k(w)}_{B^k}$. This product is an integer of bitsize at most
  $(N+1)(m_k+1)\tau$. Hence, doing this, both matrices will have entries
  where all the coefficients are integers of bitsize at most
  $(N+1)(m_k+1)\tau$ as well. Now by Proposition~\ref{PROP:DetPolyBound}
  again the bitsize of the coefficients of both determinants is at most 
  \begin{gather*}
    ( (N+1)(m_k+1)\tau + \bit(m_k) )(m_k+1) + N\bit(m_k+2) \leq \\
    2(N+1)(m_k+1)^2\tau
  \end{gather*}
  From this the claimed bound follow.
\end{proof}

We can now state the degree and separation bounds for Shapley games.
\begin{proposition}\label{PROP:sepShap}
  Let $\Gamma$ be a Shapley game with $N$ positions and $m$ actions
  for each player in each position and all rewards and transition
  probabilities being rational numbers with numerators and
  denominators of bitsize at most $\tau$, and let $v$ be the value
  vector of $\Gamma$. Then, each entry of $v$ is an algebraic number
  of degree at most $(2m+5)^N$ and the defining polynomial has
  coefficients of bitsize at most $21m^2N^2\tau
  (2m+5)^{N-1}$. Finally, if an entry of $v$ is not an integer
  multiple of $2^{-j}$, it differs from any such multiple by at least
  $2^{ -22m^2N^2\tau (2m+5)^{N-1} - j (2m+5)^{N} -1}$.
\end{proposition}

\begin{proof}
  From Theorem~\ref{THM:ShapleyPolynomialSystem} the value of $\Gamma$
  is among the isolated real solutions of a system of
  $\sum_{i=1}^{N}{2m \choose m} \leq 4^m$ polynomials, of degree at
  most $m+2$ and bitsize at most $2(N+1)(m+1)^2\tau +1 \leq
  4Nm^2\tau$.  Theorem~\ref{th:isol-real-root-bd-full} implies that
  the algebraic degree of the solutions is $(2(m+2)+1)^N = (2m+5)^N$
  and the defining polynomial has coefficients of magnitude at most
  \[
2^{(8m^2 N^2 \tau + 8N m + 5 N \lg(m)) (2m+5)^{N-1}} \leq
  2^{21m^2N^2\tau (2m+5)^{N-1}} \enspace .
\]

  For a position $k$, let the defining polynomial be $A(v_k)$. To
  compute a lower bound on the difference between a root of $A$ and a
  number $2^{-j}$, it suffices to apply the map $v_k \mapsto v_k +
  2^{-j}$ to $A$ and compute a lower bound for the roots of the
  shifted polynomial.  The shifted polynomial also has degree
  $(2m+5)^N$, but its maximum coefficient bitsize is now bounded by
  $21m^2N^2\tau (2m+5)^{N-1} + j (2m+5)^{N} + 4 \lg (2m+5)^{N} \leq
  22m^2N^2\tau (2m+5)^{N-1} + j (2m+5)^{N}$.  By applying
  Lemma~\ref{lem:uni-bounds} we get the result.
\end{proof}

\subsection{Everett Games}

Our bounds on degree, coefficient size, and separation for Everett
games are obtained by a reduction to the more general results about
the first-order theory of the reals.

\begin{theorem}
  Let $\Gamma$ be an Everett game, with $N$ positions. Assume that in
  position $k$, the two players have $m_k$ and $n_k$ actions
  available. Assume further that all payoffs and probabilities in
  $\Gamma$ are rational numbers with numerators and denominators of
  bitsize at most $\tau$.
  
  Then there is a quantified formula with $N$ free variables that
  describes whether a vector $v^*$ is the value vector of $\Gamma$.
  The formula has two blocks of quantifiers, where the first block
  consists of a single variable and the second block consists of $2N$
  variables. Furthermore the formula uses at most
  $(2N+3)+2(m+2)\sum_{k=1}^N \binom{n_k+m_k}{m_k}$ different
  polynomials, each of degree at most $m+2$ and having coefficients
  of bitsize at most $2(N+1)(m+2)^2\bit(m)\tau$, where $m =
  \max_{k=1}^N\left(\min(n_k,m_k)\right)$.
  \label{THM:EverettPolynomialSystem}
\end{theorem}
\begin{proof}
  By Theorem~\ref{THM:CriticalVectorExistence} we may express the
  value vector $v^*$ by the following first-order formula with free
  variables $v$: $ (\forall \epsilon)(\exists v_1,v_2)\ (\epsilon \leq
  0) \vee (\norm{v-v_1}^2 < \epsilon \wedge \norm{v-v_2}^2 < \epsilon
  \wedge v_1 \in C_1(\Gamma) \wedge v_2 \in C_2(\Gamma)) \enspace .  $
  Here the expressions $v_1 \in C_1(\Gamma)$ and $v_2 \in C_2(\Gamma)$
  are shorthands for the quantifier free formulas of polynomial
  inequalities implied by the definitions of $C_1(\Gamma)$ and
  $C_2(\Gamma)$. We provide the details below for the case of
  $C_1(\Gamma)$. The case of $C_2(\Gamma)$ is analogous.  By
  definition $v_1 \in C_1(\Gamma)$ means $M(v_1) \dotgeq v_1$, which
  in turn is equivalent to $\wedge_{k=1}^N ( (\val(A^k(v_1)) > v_{1k}
  \wedge v_{1k} >0 )$ $ \vee $ $ (\val(A^k(v_1)) \geq v_{1k} $ $\wedge
  $ $(v_{1k} \leq 0) ) ).$ Now we can rewrite the predicate
  $\val(A^k(v_1)) > v_{1k}$ to the following expression: $\vee_{B^k} (
  (v_1 \in F^{A^k+}_{B^k} \wedge \det((M^{A^k(v_1)}_{B^k})_{m_k+1}) >
  v_{1k}\det(M^{A^k(v_1)}_{B^k}))) \vee ((v_1 \in F^{A^k-}_{B^k}
  \wedge \det((M^{A^k(v_1)}_{B^k})_{m_k+1}) <
  v_{1k}\det(M^{A^k(v_1)}_{B^k}))),$ where the disjunction is over all
  potential basis sets, and each of the expressions $v_1 \in
  F^{A^k+}_{B^k}$ and $v_1 \in F^{A^k-}_{B^k}$ are shorthands for the
  conjunction of the $m_k+1$ polynomial inequalities describing the
  corresponding sets.

By a similar
analysis as in the proof of Theorem~\ref{THM:ShapleyPolynomialSystem}
we get the following bounds, assuming without loss of generality that
$m_k \leq n_k$:
  The predicates $v_1 \in F^{A^k+}_{B^k}$ and $v_1 \in F^{A^k-}_{B^k}$
  can be written as a quantifier free formulas using at most $m_k+1$
  different polynomials, each of degree at most $m_k+2$ and having
  coefficients of bitsize at most $2(N+1)(m_k+2)^2\bit(m_k)\tau$.
Also,
  the predicate $\val(A^k(v_1)) > v_{1k}$ can be written as a
  quantifier free formula using at most $(m_k+2)\binom{n_k+m_k}{m_k}$
  different polynomials, each of degree at most $m_k+2$ and having
  coefficients of bitsize at most $2(N+1)(m_k+2)^2\bit(m_k)\tau$.

  Combining these further, for all positions we have the following
  statement (that shall be used also in our upper bound for strategy
  iteration for concurrent reachability games in
  Section~\ref{sec-si}).
\begin{lemma}
  \label{LEM:QFreeC1C2}
  The predicate $v_1 \in C_1(\Gamma)$ can be written as a quantifier
  free formula using at most $\sum_{k=1}^N
  1+(m+2)\binom{n_k+m_k}{m_k}$ different polynomials, each of degree
  at most $m+2$ and having coefficients of bitsize at most
  $2(N+1)(m+2)^2\bit(m)\tau$, where $m =
  \max_{k=1}^N\left(\min(n_k,m_k)\right)$.
\end{lemma}
From this the statement of the theorem easily follows.
\end{proof}

We shall also need the following basic statement about univariate
representations.

\begin{lemma}\label{LEM:convert}
  Let $\alpha$ be a root of $f \in \ZZ[x]$, which is of degree $d$ and
  maximum coefficient bitsize at most $\tau$.  Moreover, let
  $g(x) = p(x)/q(x)$ where $p, q \in \ZZ[x]$ are of degree
  at most $d$, have maximum coefficient bitsize at most $\tau$, 
  and $q(\alpha) \not= 0$.
  Then the minimal polynomial of
  $g(\alpha)$ is a univariate polynomial of degree at most $2d$ and
  maximum coefficient bitsize at most $2 d \tau + 7 d \lg{d}$.
\end{lemma}
\begin{proof}
  The minimal polynomial of $g(\alpha)$ is among the square-free
  factors of the following (univariate) resultant with respect to $y$:
  \begin{displaymath}
    r(x) = \mathrm{res}_y( f(y), q(y)x - p(y)) \in \ZZ[x].
  \end{displaymath}
  The degree of $r$ is bounded by $d$ and its maximum coefficient
  bitsize is at most $2 d \tau + 5 d \lg{d}$
  \cite[Proposition~8.46]{BasuPollackRoy2006}. Any factor of $r$ has
  maximum coefficient bitsize at most $2 d \tau + 7 d \lg{d}$, due to
  the Landau-Mignotte bound, see, e.g., Mignotte \cite{Mign91}.
\end{proof}

We can now apply the machinery of semi-algebraic geometry to get the
desired bounds on degree and the separation bounds.

\begin{proposition}\label{PROP:sepEve}
  Let $\Gamma$ be an Everett game with $N$ positions, $m$ actions for
  each player in each position, and rewards and transition
  probabilities being rational numbers with numerators and
  denominators of bitsize at most $\tau$, and let $v$ be the value
  vector of $\Gamma$. Then, each entry of $v$ is an algebraic number
  of degree at most $m^{O(N^2)}$ and the defining polynomial has
  coefficients of bitsize at most $\tau m^{O(N^2)}$.  Finally, if an
  entry of $v$ is not a multiple of $2^{-j}$, it differs from any such
  multiple by at least $2^{- \max\{\tau,j\} \,m^{O(N^2)}}$.
\end{proposition}
\begin{proof}
  We use Theorem 14.16 (Quantifier Elimination) of Basu, Pollack and
  Roy \cite{BasuPollackRoy2006} on the formula of
  Theorem~\ref{THM:EverettPolynomialSystem} to find a quantifier free
  formula expressing that $v$ is the value vector of the game. Next,
  we use Theorem 13.11 (Sampling) of \cite{BasuPollackRoy2006} to this
  quantifier free formula to find a univariate representation of the
  value vector $v$ satisfying the formula from
  Lemma~\ref{THM:EverettPolynomialSystem}. That is, we obtain polynomials
  $f,g_0,\dots,g_{2N}$, with $f$ and $g_0$ coprime, such that $v =
  (g_1(t)/g_0(t),\dots,g_{2N}(t)/g_0(t))$, where $t$ is a root of
  $f$. These polynomials are of degree $m^{O(N^2)}$ and their
  coefficients have bitsize $\tau m^{O(N^2)}$. We apply
  Lemma~\ref{LEM:convert} to the univariate representation to obtain
  the desired defining polynomials. Finally, we obtain the separation
  bound using Lemma~\ref{lem:uni-bounds} in the same way as in the proof of Proposition~\ref{PROP:sepShap}
\end{proof}


\subsection{Gillette's Stochastic Games}
Our bounds on degree, coefficient size, and separation for Gillette
games are obtained, as in the case of Everett games but in a more
involved way, by a reduction to the more general results about the
first-order theory of the reals.

\label{sec-gil}
\begin{theorem}
  Let $\Gamma$ be a Gillette game, with $N$ positions. Assume that in
  position $k$, the two players have $m_k$ and $n_k$ actions
  available. Assume further that all payoffs and probabilities in
  $\Gamma$ are rational numbers with numerators and denominators of
  bitsize at most $\tau$.
  
  Then there is a quantified formula with $N$ free variables that
  describes whether a vector $v^*$ is the value vector of $\Gamma$. The
  formula has four blocks of quantifiers, where the first three blocks
  consists of a single variable and the fourth block consists of $N$
  variables.  Furthermore the formula uses at most
  $4+2(m+2)\sum_{k=1}^N \binom{n_k+m_k}{m_k}$ different polynomials,
  each of degree at most $2(m+2)$ and having coefficients of bitsize
  at most $2(N+1)(m+2)^2\bit(m)\tau$, where $m =
  \max_{k=1}^N\left(\min(n_k,m_k)\right)$.
  \label{THM:MertensNeymanPolynomialSystem}
\end{theorem}
\begin{proof}
  By Theorem~\ref{THM:MertensNeyman} we may express the value vector
  $v^*$ by the following first-order formula with free variables $v$.
\[
(\forall \epsilon>0)(\exists \lambda_\epsilon>0)(\forall \lambda,
0<\lambda\leq\lambda_\epsilon)(\exists v') 
(v'=\lambda\val(\Gamma_\lambda) \wedge 
\norm{v'-v}^2 < \epsilon)
 \enspace .
\]
Here $\Gamma_\lambda$ is the $(1-\lambda)$-discounted version of
$\Gamma$, and the expression $v'=\lambda\val(\Gamma_\lambda)$ is a
shorthand for a quantifier free formula of polynomial equalities and
inequalities expressing that $v'$ is the normalized vector of values
of $\Gamma_\lambda$, and may be expressed as
\[
\bigwedge_{k=1}^N \left(
\bigvee_{B^k} \left( \left(v' \in F^{A_\lambda^k+}_{B^k}  \vee v' \in F^{A_\lambda^k-}_{B^k} \right) 
\wedge \det((M^{A_\lambda^k(v')}_{B^k})_{m_k+1}) = \lambda v'_{k}\det(M^{A_\lambda^k(v')}_{B^k}) \right) 
\right) \enspace ,
\]
using Theorem~\ref{THM:ShapleyValue} and where $A_\lambda^k$ is the
parameterized matrix game corresponding to $\Gamma_\lambda$ obtained
as explained in Section~\ref{def-gil}. Here, as in
the last section, the disjunction is over all potential basis sets,
and each of the expressions $v' \in F^{A_\lambda^k+}_{B^k}$ and $v'
\in F^{A_\lambda^k-}_{B^k}$ are shorthands for the conjunction of the
$m_k+1$ polynomial inequalities describing the corresponding sets.

We next analyze the bounds in the following. By a similar analysis as
in the proof of Theorem~\ref{THM:EverettPolynomialSystem} and
Theorem~\ref{THM:ShapleyPolynomialSystem} we get the following bounds,
assuming without loss of generality that $m_k \leq n_k$.

\begin{lemma}
  The predicates $v' \in F^{A_\lambda^k+}_{B^k}$ and $v' \in
  F^{A_\lambda^k-}_{B^k}$ can be written as a quantifier free formulas
  using at most $m_k+1$ different polynomials, each of degree at most
  $2(m_k+2)$ and having coefficients of bitsize at most
  $2(N+1)(m_k+2)^2\bit(m_k)\tau$.
\end{lemma}

The larger degree compared to the case of Everett games is due to the additional variable
$\lambda$. The same is true for the remaining predicate, hence we
obtain the following.

\begin{lemma}
  The predicate $v'=\lambda\val(\Gamma_\lambda)$ can be written as a
  quantifier free formula using at most $\sum_{k=1}^N
  (m+2)\binom{n_k+m_k}{m_k}$ different polynomials, each of degree at
  most $2(m+2)$ and having coefficients of bitsize at most
  $2(N+1)(m+2)^2\bit(m)\tau$, where $m =
  \max_{k=1}^N\left(\min(n_k,m_k)\right)$.
\end{lemma}

From this the statement easily follows.
\end{proof}

Proceeding exactly as in the proof of Proposition~\ref{PROP:sepEve},
we may now prove the following proposition, giving the exact same
statement for Gillette games as for Everett games. Note, however, that
since more blocks of quantifiers have to be eliminated, the constants
in the big-O's are likely worse.
\begin{proposition}\label{PROP:sepGil}
  Let $\Gamma$ be a Gillette game with $N$ positions, $m$ actions for
  each player in each position, and payoffs and transition
  probabilities being rational numbers with numerators and
  denominators of bitsize at most $\tau$, and let $v$ be the value
  vector of $\Gamma$. Then, each entry of $v$ is an algebraic number
  of degree at most $m^{O(N^2)}$, and the defining polynomial has
  coefficients of bitsize at most $\tau m^{O(N^2)}$.  Finally, if an
  entry of $v$ is not a multiple of $2^{-j}$, it differs from any such
  multiple by at least $2^{- \max\{\tau,j\} \,m^{O(N^2)}}$.
\end{proposition}

Next we will obtain a bound on the discount factor for guaranteeing a
sufficient approximation of the undiscounted game by the discounted
one.  We will consider the same formula, strip away the first two
quantifiers, replacing the variable $\epsilon$ by a fixed constant and
letting $\lambda_\epsilon$ be a free variable. Next binding the
previous free variables $v$ and expressing that these take the values
of the value vector of $\Gamma$ we in effect obtain a first order
formula expressing a sufficient condition for whether a given discount
factor $\gamma=1-\lambda$ ensures that the values vectors of $\Gamma$
and $\Gamma_\lambda$ are $\epsilon$-close in every coordinate.

\begin{theorem}\label{gilbound}
  Let $\Gamma$ be a Gillette game, with $N$ positions. Assume that in
  position $k$, the two players have $m_k$ and $n_k$ actions
  available. Assume further that all payoffs and probabilities in
  $\Gamma$ are rational numbers with numerators and denominators of
  bitsize at most $\tau$.

  Let $\epsilon=2^{-j}$. Then there is a quantified formula with one
  free variable that gives a sufficient condition for whether a given
  discount factor $\gamma=1-\lambda_\epsilon$ guarantees that
  $\norm{\val(\Gamma)-\lambda_\epsilon\val(\Gamma_{\lambda_\epsilon})}^2
  \leq \epsilon$.

  The formula has five blocks of quantifiers, where the first block
  consists of $N$ variables, second of 1 variable, third and fourth of
  2 variables and the fifth of $2N$ variables.  Furthermore the
  formula uses at most $6+4(m+2)\sum_{k=1}^N \binom{n_k+m_k}{m_k}$
  different polynomials, each of degree at most $2(m+2)$ and having
  coefficients of bitsize at most
  $\max\{j,2(N+1)(m+2)^2\bit(m)\tau\}$, where $m =
  \max_{k=1}^N\left(\min(n_k,m_k)\right)$.
  \label{THM:MertensNeymanPolynomialSystem2}
\end{theorem}
\begin{proof}
  Following the proof of
  Theorem~\ref{THM:MertensNeymanPolynomialSystem} above, we may
  express the condition by the following first-order formula with free
  variable $\lambda_\epsilon$.
  \[
  (\exists v) (\forall \lambda, 0<\lambda\leq\lambda_\epsilon)(\exists
  v') 
  (v'=\lambda\val(\Gamma_\lambda) \wedge \norm{v'-v}^2 < \epsilon
  \wedge v=\val(\Gamma))
  \enspace ,
  \]
  and letting $v=\val(\Gamma)$ be a shorthand for the entire formula
  guaranteed by Theorem~\ref{THM:MertensNeymanPolynomialSystem}. We
  obtain the formula as claimed by converting the above formula into
  prenex normal form. The rest of the analysis follows closely the
  proof of Theorem~\ref{THM:MertensNeymanPolynomialSystem} and is
  hence omitted.
\end{proof}

We can now apply again the machinery of semi-algebraic geometry to get
a bound on $\lambda_\epsilon$ above as a function of $\epsilon$.

\begin{proposition}\label{PROP:Gillette-discountbound}
  Let $\Gamma$ be a Gillette game with $N$ positions, $m$ actions for
  each player in each position, and payoffs and transition
  probabilities being rational numbers with numerators and
  denominators of bitsize at most $\tau$, and let
  $\epsilon=2^{-j}$. Then there exists $\lambda_\epsilon =
  \epsilon^{\tau m^{O(N^2)}}$, such that
  $\norm{\val(\Gamma)-\lambda_\epsilon\val(\Gamma_{\lambda_\epsilon})}^2
  \leq \epsilon$.
\end{proposition}
\begin{proof}
  First we use Theorem 14.16 (Quantifier Elimination) of Basu et
  al.\cite{BasuPollackRoy2006} to the formula of
  Theorem~\ref{gilbound} to obtain an equivalent quantifier free
  formula. The (univariate) polynomials in this formula are of degree
  $m^{O(N^2)}$ and has coefficients of bitsize
  $\max\{\tau,j\}m^{O(N^2)}=\log(1/\epsilon)\tau m^{O(N^2)}$. We can
  then again use Theorem 13.11 (Sampling) of \cite{BasuPollackRoy2006},
  Lemma~\ref{LEM:convert}, and Lemma~\ref{lem:uni-bounds} to obtain
  the lower bound $\lambda_\epsilon = \epsilon^{\tau m^{O(N^2)}}$.
\end{proof}


\section{Degree and separation bounds for isolated real solutions}
\label{sec-isolated}

In this section we prove general results about the coordinates of
isolated solutions of polynomial systems. The result as stated below
provides concrete bounds on the algebraic degree, coefficient size and
separation.  
\begin{theorem}
  \label{th:isol-real-root-bd-full}
  Consider a polynomial system of equations
  \begin{equation}
    (\Sigma) \quad \quad
    g_1(x_1, \dots, x_n) = \cdots = g_m(x_1, \dots, x_n) = 0 \enspace,
  \label{eq:orig-system}
  \end{equation}
  with polynomials of degree at most $d$ and integer
  coefficients of magnitude at most $2^{\tau}$.

  Then, the coordinates of any {\em isolated} (in Euclidean topology) real 
  solutions of the system are
  real algebraic numbers of degree at most $(2d+1)^n$, and their
  defining polynomials have coefficients of magnitude at most 
  $2^{2n(\tau+4n\lg(dm))(2d+1)^{n-1}}$.
  Also, if $\gamma_j =
  (\gamma_{j,1}, \cdots, \gamma_{j,n})$ is an isolated solution of $(\Sigma)$,
  then for any $i$, either
  \begin{equation}
    2^{-2n(\tau + 2n\lg(dm))(2d+1)^{n-1}} < \abs{\gamma_{j,i }} 
    \quad \text{ or } \quad \gamma_{j,i} = 0 \enspace.
    \label{eq:isol-rr-lower-bd}
  \end{equation}
  Moreover, given coordinates of isolated solutions of two such systems, if they are not identical, they differ by at least
  \begin{equation}
    \sep(\Sigma) 
    \geq 2^{-3n(\tau + 2n\lg(dm))(2d+1)^{2n-1} - \frac{1}{2}\lg(n)}
    \label{eq:isol-rr-sep-bd}
    \enspace .
  \end{equation}
\end{theorem}

Before the proof of the theorem we will need to establish some preliminary results.

\subsection{Isolated solutions, minimizers and the $u$-resultant}

We will use ideas from \cite{HanzonJibetean} used for for global
minimization of polynomial functions in order to reach an appropriate
system to analyze. The solutions of the system ($\Sigma$), which
consists of real polynomials of total degree at most $d$, are the
minimizers of the polynomial
\begin{equation}
  \label{minG}
  G(x_1, \dots, x_n) = g_1(x_1,\dots, x_n)^2 + \cdots + g_m(x_1, \dots, x_n)^2
\end{equation}
in $\RR^n$. Furthermore, if $z$ is an isolated real solution of
$(\Sigma)$, then $z$ is an isolated minimizer for
\eqref{minG}.  Let $G_i(\x) = \frac{\partial G(\x)}{\partial x_i}$.
The critical points of $G(\x)$ are among the solution set of the
system
    \begin{equation}
      \label{eq:G-sys}
            G_1(\x) = \dots = G_n(\x) = 0.
    \end{equation}
    
If the number solutions of the system above is finite, then we can
use the multivariate resultant%
\footnote{Following closely \cite{c-crmp-87}, for $n$ homogeneous
  polynomials $f_1, \dots, f_n$ in $n$ variables $x_1, \dots, x_n$, of
  degrees $d_1, \dots, d_n$ respectively, the multivariate resultant
  is a single polynomial in the coefficient of $f_i$, the vanishing of
  which is the necessary and sufficient condition for the polynomials
  to have a common non-trivial solution in the algebraic closure of
  the field of their coefficients. The resultant is of degree $d_1d_2
  \cdots d_{i-1}d_{i+1} \cdots d_n$ in the coefficients of $f_i$.}
\cite{CoxLittleOShea98,c-crmp-87} to compute them.
We homogenize the polynomials using a new variable $x_{0}$ and
introduce the linear form $G_0= u_0 x_0 + u_1 x_1 + \cdots + u_n
x_n$. We then compute the multivariate resultant of $G_1, \dots, G_n$
and $G_0$ with respect to the variables $x_0, x_1, \dots, x_n$, and a
homogeneous polynomial with degree equal to the product of the degrees
of $G_i$ is obtained. This is called the $u$-resultant \cite{Wae}, see
also \cite{c-crmp-87}.  If the number of solutions is finite then the
resultant is non-vanishing for almost all linear forms $G_0$, and if
we factorize it to linear forms over the complex numbers then we can
recover the solutions of the system.

    To compute, or as in our case to bound, the $\ell$-th coordinates
    of the solution set, we may assume $u_{\ell}=-1$ and $u_i=0$, for all
    $i$ different from 0 and $\ell$. Then the $u$-resultant is a univariate
    polynomial in $u_{0}$, and its solutions correspond to the
    $\ell$-th coordinates of the solutions of the system.

    However, the multivariate resultant vanishes identically if the
    system has an infinite number of solutions. This is the case when
    the variety has positive dimension or, simply, the variety has a
    component of positive dimension at infinity, also known as {\em
      excess component}.

\subsection{Gr\"obner bases and Deformations}

First we recall the following fundamental results from the theory of
Gr\"obner bases. Let $k$ be a field and $R = k[x_1, \dots, x_n]$.  For
an extension field $K\supset k$ and an ideal $I\subset R$ we let
$V_K(I) := \{x\in K^n \mid f(x) = 0, \, \forall f\in I\}$.
\begin{lemma}
  \label{Lemmagb}
  Consider an ideal $I\subset R$, such that
  $ d:= \dim_k R/I < \infty$.
  \begin{enumerate}[(i)]
  \item\label{gbi}
    If $(z_1, \dots, z_n)\in V_K(I)$. Then $z_j\in K$ is
    algebraic over $k$ of degree at most $d$.
  \item\label{gbii}
    Suppose that $I = (f_1, \dots, f_n)$ with
    \begin{align*}
      f_1(\x) &= x_1^{d_1} + h_1(\x)\\
      &\vdots \\
      f_n(\x) &= x_n^{d_n} + h_n(\x),
    \end{align*}
    where $\deg(h_j) < d_j$. Then $\dim_k R/I = d_1 \cdots d_n$.
  \end{enumerate}
\end{lemma}
Here item (i) follows from the proof of Theorem 6, Chapter 5 of
\cite{CoLiOS92} (more precisely, the proof of (v) $\Rightarrow$
(i)). Item (ii) follows from Proposition 4, also from Chapter 5 of
\cite{CoLiOS92}, noting that $(f_1,\dots,f_n)$ is a Gr\"obner basis
with respect to the graded lexicographic order.

Next, in order to apply the $u$-resultant as described above, we will
symbolically \emph{perturb} the system. We need to do it in such a way
that the perturbed system becomes $0$-dimensional and also that from
the solutions of this perturbed system we can recover the isolated
real solutions of the original system. In \cite{HanzonJibetean} the
deformation
\[
G_\lambda(x) = G(x) + \lambda (x_1^{2(d+1)} + \cdots + x_n^{2(d+1)}),
\]
where $\lambda>0$ is introduced. By Lemma~\ref{Lemmagb}\eqref{gbii}
\[
\dim_\RR R/\nabla I(G_\lambda) \leq (2d + 1)^n
\]
for $\lambda>0$, where $\nabla I(G)$ is the gradient ideal
$(\frac{\partial G_\lambda}{\partial x_1}, \dots, \frac{\partial
  G_\lambda}{\partial x_n})$.  Let
\[
X_\lambda = V(\nabla I(G_\lambda))\subset \RR^n. 
\]
Notice that $\abs{X_\lambda}\leq
\dim_k R/\nabla I(G_\lambda) = (2 d + 1)^n$. We wish to reason about
the ``limit'' $L = \lim_{\lambda\to 0} X_\lambda$.  To make this more
precise we define
\[
L = \{x\in \RR^n \mid \forall\epsilon>0\,\exists \lambda_\epsilon>0:
B(x, \epsilon)\cap X_\lambda\neq \emptyset,\, \text{for every
}\lambda\, \text{with}\, 0 < \lambda < \lambda_\epsilon\}.
\]
It is rather difficult to decide if a given point is in $L$. For one
thing the polynomial system may have several bigger components not
related to the limit. In our case, we have the following result, which
allows us to recover the real solution if we solve the system in the
limit, that is as $\lambda \rightarrow 0$.

\begin{proposition}\label{propinL}
  If $z = (z_1, \dots, z_n)$ is an isolated solution of
  $(\Sigma)$, eq.~(\ref{eq:orig-system}), then $z\in L$.
\end{proposition}
\begin{proof}
  By the isolation of $z$ there exists $\delta>0$, such that $G(x)>0$
  for every $x\in B(z, \delta)\setminus \{z\}$. Therefore $m = \min
  \{G(x) \mid x\in \partial B(z, \delta)\} > 0$. Pick $\lambda > 0$ so
  that
  \[
  G_\lambda(z) = \lambda (z_1^{2(d+1)} + \cdots + z_n^{2(d+1)}) < m
  \]
  Since
  \[
  m
  \leq \min\{G_\lambda(x)\mid x\in \partial B(z, \delta)\},
  \]
  we know that the minimum of $G_\lambda$ on $B(z, \delta)$ is
  attained in $B(z, \delta)^\circ$.  Thus, $X_\lambda\cap B(z,
  \delta)\neq \emptyset$.
\end{proof}

\subsection{Proof of Theorem~\ref{th:isol-real-root-bd-full}}
For the proof of Theorem \ref{th:isol-real-root-bd-full} we additionally
need the following fundamental bounds.
\begin{lemma}
  \cite{BasuPollackRoy2006,Mign91,Yap2000}
  \label{lem:uni-bounds}
  Let $f \in \ZZ[x]$ of degree $d$, then for any non-zero root $\gamma$ it holds
  \begin{displaymath}
    (2 \norm{f}_{\infty})^{-1} \leq   \abs{\gamma} \leq 2 \norm{f}_{\infty} \enspace.
  \end{displaymath}
  If $\sep{f}$ is the separation bound, that is the minimum distance between the roots, then
  \begin{displaymath}
    \sep{f} = \min_{i\not=j}\abs{\gamma_i - \gamma_{j}} \geq
    d^{-(d+2)/2} \norm{f}_2^{1-d} \enspace .
  \end{displaymath}
\end{lemma}

\begin{proof}[Proof of Theorem~\ref{th:isol-real-root-bd-full}]
  Let $\gamma_j = (\gamma_{j,1}, \cdots, \gamma_{j,n})$ be isolated
  real solutions of the system $(\Sigma)$. As
  above, we consider
\[
  G(x_1, \dots, x_n) = g_1(x_1,\dots, x_n)^2 + \cdots + g_m(x_1, \dots, x_n)^2
\]
and its pertubation
\[
G_\lambda(x) = G(x) + \lambda (x_1^{2(d+1)} + \cdots + x_n^{2(d+1)}),
\]
Form the system of partial derivatives 
\[
f_i = G_i + (2d+2)\lambda x_i^{2d+1} \enspace ,
\]
where $G_i(\x) = \frac{\partial G(\x)}{\partial x_i}$. We homogenize
the polynomials using a new variable $x_0$ and introduce the linear
form $u_0x_0 + \dots + u_nx_n$ specialized to the $l$th coordinate as
describe above. That is we add the polynomial
\[
f_0 = ux_0 - x_1
\]
Let the resulting system be $(\Sigma_0)$.

For a polynomial $f$, let $\bitsize{f}$ be the maximum coefficient
bitsize, that is $\bitsize{f} = \ceil{ \lg{ \norm{f}_{\infty}}}$. We
have $\dg{G} \leq 2d$ and $\bitsize{G}\leq 2\tau+2n\lg(d m)$. Write $G_i$ on the form
\[
  G_i(\x)  =\sum_{j=1}^{2d-1}{ c_{i, j} { \x}^{a_{i, j}}} \in \ZZ[\x],
\]
where $1 \leq i \leq n$, and let $\mathbf{c}$ be the set of all
coefficients $c_{i,j}$.  It holds that $\dg{G_i} = 2d-1$,
$\norm{G_i}_{\infty} \leq 2d \norm{G}_{\infty}$.

Let $D=(2d+1)^n$ and $D_1=(2d+1)^{n-1}$. For the system $(\Sigma_0)$
we consider the multivariate resultant $R$ in the variables $x_0, x_1,
\dots, x_n$.  It is a polynomial in the coefficients of $G$, $u$ and
$\lambda$, that is $R \in (\ZZ[\mathbf{c}, \lambda])[u]$,
\cite{CoxLittleOShea98}.  It has degree $D_1$ in the coefficients of
$G_i$, where $1 \leq i \leq n$, and degree $D$ in the coefficients of
$G_0$, which are 1 and $u$.  To be more specific, $R$ is of the form
\[
R = \dots + {\varrho}_k \, u^k \, {\widetilde
  {\mathbf{c}}_{1,k}}^{D_1} \, {\widetilde {\mathbf{c}}_{2,k}}^{D_1}
\, \cdots {\widetilde {\mathbf{c}}_{n,k}}^{D_1} \, + \dots ,
\]
where $\varrho_k\in\ZZ$, and ${\widetilde {\mathbf{c}}_{i,k}}^{D_1}$
is of the form $\lambda^{\mu} {\mathbf{c}_{i,k}}^{D_1-\mu}$ where the
second factor corresponds to a monomial in the coefficients $c_{ij}$,
of total degree $D_1-\mu$, for some $\mu$ smaller than $D_1$.

The lowest-degree nonzero coefficient of $R$, $R_u$, seen as
univariate polynomial in $\lambda$, is a projection operator: it
vanishes on the projection of any 0-dimensional component of the
algebraic set defined by $(\Sigma_0)$
\cite{c-crmp-87,de-cm-2001,emt-issac-2010}.  In our case the $\ell$-th
coordinates of the isolated solutions of (\ref{eq:G-sys}) are among
the roots of $R_u$.

It holds that $R_u \in \ZZ[\mathbf{c}][u]$, and $\dg{R_u}\leq D$.
Notice that the bound on the degree of $R_u$, that is $D=(2d+1)^{n}$,
is also an upper bound on the algebraic degree on the coordinates of
the solutions of (\ref{eq:orig-system}).  Which proves the first
assertion of the theorem.

To compute the bounds on the roots of $R_u$, and thus bounds on the
isolated solutions of (\ref{eq:G-sys}), we should bound the magnitude
of its coefficients. For the latter, it suffices to bound the
coefficients of $R$. Let
\[
\norm{R}_{\infty} \leq \max_{k} \abs{ {\varrho}_k \, {
    {\mathbf{c}}_{1,k}}^{D_1} \, { {\mathbf{c}}_{2,k}}^{D_1} \, \cdots
  { {\mathbf{c}}_{n,k}}^{D_1} } \leq \max_{k} \abs{ {\varrho}_k }
\cdot \max_{k} \abs{ { {\mathbf{c}}_{1,k}}^{D_1} \, {
    {\mathbf{c}}_{2,k}}^{D_1} \, \cdots { {\mathbf{c}}_{n,k}}^{D_1} }
= h \cdot C \enspace .
\]
To bound $\varrho_k$ we need a bound on the number of integer points
of the Newton polygons of $f_i$ \cite{sombra-ajm-2004}, which we
denote by $(\#Q_i)$. We refer to \cite{emt-issac-2010} for details.
For all $k$ we have
\[
\abs{\varrho_k} \leq h = (n+1)^D \prod_{i=1}^{n}(\#Q_i)^{D_1} \leq
2^{n D_1} D^{n D_1} \enspace .
\]
Moreover
\[
\max_{k} \abs{ { {\mathbf{c}}_{1,k}}^{D_1} \, {
    {\mathbf{c}}_{2,k}}^{D_1} \, \cdots { {\mathbf{c}}_{n,k}}^{D_1} }
= \prod_{i=1}^{n}{ \norm{G_i}_{\infty}^{D_1}} \leq \Paren{d
  \norm{G}_{\infty}}^{n D_1} = C \enspace .
\]
Hence
\[
\norm{R_u}_{\infty} \leq \norm{R}_{\infty} \leq h C = (2D d \norm{G}
)^{n D_1} \leq 2^{2n(\tau + 2n\lg(dm))(2d+1)^{n-1}} \enspace.
\]
Using Cauchy's bound (Lemma~\ref{lem:uni-bounds}) any of the non-zero
roots $\gamma_{j,i}$ of $R_u$ satisfies
\[
\abs{ \gamma_{j,i}} > \norm{R_u}_{\infty}^{-1} \geq (hC)^{-1} \geq
2^{-2n(\tau + 2n\lg(dm))(2d+1)^{n-1}} \enspace .
\]
Notice that the defining polynomial of $\gamma_{j,i}$ is the square-free part of $R_u$, 
which has bitsize at most 
$2n(\tau + 2n\lg(dm))(2d+1)^{n-1} + (2d+1)^{n-1} + 2\lg(2d+1)^{n-1}
\leq 2n(\tau + 4n\lg(dm))(2d+1)^{n-1}$.

To bound the minimum distance between the isolated roots of
$(\Sigma)$, we notice that
\[
\sqrt{n} \, \sep(\Sigma) \geq \sqrt{n} \, \min_{i \not= j}
\norm{\gamma_{i} - \gamma_{j}}_{\infty} \geq \min_{i \not= j}
\norm{\gamma_{i} - \gamma_{j}}_{2} \geq \min_{i \not= j}
\abs{\gamma_{i,\ell} - \gamma_{j,\ell}},
\]
for any $1 \leq \ell \leq n$ and where the last minimum is considered
over all $\gamma_{i,\ell} \not= \gamma_{j,\ell}$.

Using the separation bound for univariate polynomials
(Lemma~\ref{lem:uni-bounds}), we get
\[
\sep(R_u) = \min_{i \not= j} \abs{\gamma_{i,\ell} - \gamma_{j,\ell}}
\geq D^{-\frac{D+2}{2}} \norm{R_u}_2^{1-D} \geq D^{-\frac{D+2}{2}}
(\sqrt{D} \norm{R_u}_{\infty})^{1-D} ,
\]
and so
\[
\sep(R_u) = \min_{i \not= j} \abs{\gamma_{i,\ell} - \gamma_{j,\ell}}
\geq 2^{-3n(\tau + 2n\lg(dm))(2d+1)^{2n-1}} \enspace.
\]
Finally
\[
\sep(\Sigma) \geq \sep(R_u)/\sqrt{n} \geq 2^{-3n(\tau +
  2n\lg(dm))(2d+1)^{2n-1} - \frac{1}{2}\lg(n)} \enspace .
\]
This completes the proof.
\end{proof}

Better bounds should be possible for the algebraic degree of
Theorem~\ref{th:isol-real-root-bd-full}, based for example on
Oleinik-Petrovskii, Milnor-Thom's \cite{m-pams-1964,t-pinc-1965} bound
for the sum of Betti numbers of a set of real zeros of a polynomial
system, or on improved estimates by Basu \cite{b-dcg-2003} on
individual Betti numbers; see also \cite{br-jsc-2010}. This should
lead to improved separation bounds, if used in conjunction with neat
deformation techniques and bounds on parametric Gr\"obner basis,
e.g.~\cite{br-jsc-2010,jp-jsc-2010}, and/or bounds based on the
Generalized Characteristic Polynomial and sparse multivariate
resultants \cite{canny-gcp-1990,emt-issac-2010}.  Nevertheless, it is
not possible to beat the single exponential nature of the bound, and
only improvements in the constants involved are expected.

\section{Degree lower bounds for values of Shapley games}
\label{sec-low}
In this section we give a construction of a Shapley game $\Gamma_{N,m}$
with $N+1$ positions each having at most $m$ actions,
such that the algebraic degree of the value of one of the positions is
at least $m^N$. 

Previously, Etessami and Yannakakis \cite{LMCS:EY08} gave a reduction
from the so-called square-root sum problem to the quantitative
decision problem of Shapley games. In fact from this reduction one can
obtain a Shapley game with $N$ positions where the algebraic
degree of the value of one of the positions is $2^{\Omega(N)}$.

Our results below can be viewed as a considerable extension of this,
showing how the number of actions can affect the algebraic
degree. Comparing with the upper bound $m^{O(N)}$ shows that our
result is close to optimal. The idea of the game we construct is very
simple. The game consists of a dummy game position that just gives
rise to a probability distribution over the remaining $N$
positions. Each of the remaining $N$ positions are by themselves
independent Shapley games consisting of a single position with $m$
actions. We will construct these $N$ games in such a way that their
values are independent algebraic numbers each of degree $m$. Then a
suitable linear combination of these, corresponding to the probability
distribution, will cause the dummy position to have a value which is
an algebraic number of degree $m^N$.

Actually implementing this approach seems to bring significant
challenges when $m>2$. However using the powerful Hilbert's
irreducibility theorem we are able to give a simple existence proof of
a Shapley game with the properties as stated above. Next, we will also
give an explicit proof of existence using elementary but more involved
arguments.

\subsection{The single position game}

Let $\alpha_1,\dots,\alpha_m > 0$ be arbitrary positive numbers and $0
\leq \beta < 1$. Consider the Shapley game $\Gamma(\alpha,\beta)$
consisting of a single position where each player has $m$ actions, and
the payoffs are $a_{ii} = \alpha_i$ and $a_{ij} = 0$ for $i\neq j$,
and transition probabilities $p^{11}_{ii} = \beta$ and $p^{11}_{ij} =
0$ for $i \neq j$. Thus to $\Gamma(\alpha,\beta)$ corresponds the
parameterized matrix game given by the diagonal matrix $\diag(\alpha_1
+ \beta v, \dots, \alpha_m + \beta v)$. 

By Theorem~\ref{THM:ShapleyValue}, and since the game is given by a
diagonal matrix with strictly positive entries on the diagonal, we
find that the value of the game $v$ satisfies the equation
\begin{equation}\sum_{i=1}^m \frac{v}{\alpha_i + \beta v} = 1 \enspace .
\label{EQ:degreelowerbound}
\end{equation}

More precisely, consider a diagonal matrix game $\diag(a_1,\dots,a_m)$
with strictly positive entries $a_1,\dots,a_m > 0$ on the diagonal,
and let $p$ and $q$ be optimal strategies for the row and column
player, respectively, and let $v>0$ be the value of the game. Firstly,
all $p_i>0$ as otherwise the column player could ensure payoff 0 by
playing strategy $i$. Thus $v=a_i q_i$ for all $i$, and hence also
$q_i>0$ for all $i$. But then similarly we have $v=a_i p_i$ for all
$i$. Rearranging to $p_i= v/a_i$ and doing summation over $i$ gives
the claimed equation.

Define the polynomial $f_m(v) = \prod_{i=1}^m (\alpha_i + \beta v)$. Then
$f_m'(v) = \beta \sum_{i=1}^m  \prod_{j \neq i} (\alpha_j+\beta v)$. Multiplying by
$f_m(v)$ on both sides of equation~\ref{EQ:degreelowerbound} we obtain
the following.
\[
f_m(v) = v \sum_{i=1}^m \prod_{j \neq i} (\alpha_j + \beta_j v) = \frac{1}{\beta}v f_m'(v) \enspace .
\]
In the following we will specialize $\beta = 1/c$, for some $c>1$. 
We then obtain that $v$ is a root in the univariate
polynomial
\[
F_m(v) = f_m(v) - c v f'_m(v) \enspace .
\]
\subsection{Existence using Hilbert's irreducibility theorem}

We next present the simple existence proof using (a version) of
Hilbert's irreducibility theorem.

\begin{lemma}
\label{LEM:MultivariateIrreducible}
If $c > 1$ is rational, then
\[
F_m(v,\alpha_1^2,\dots,\alpha_m^2) \in \QQ[v, \alpha_1, \dots, \alpha_m]
\]
is irreducible as a multivariate polynomial in $v, \alpha_1, \dots, \alpha_m$.
\end{lemma}
\begin{proof}
  This uses induction on $m$. For $m = 1$ we have $F_1 = (1+1/c)v +
  \alpha_1^2$ which is irreducible in $\QQ[v, \alpha_1]$.  The
  induction step proceeds as follows.

\begin{align*}
  F_m &= f_m - c v f_m' = (\alpha_m^2 + \frac{1}{c} v) f_{m-1} - c v
  \frac{d}{d v}\left( (\alpha_m^2 + \frac{1}{c} v) f_{m-1}\right)\\
&= f_{m-1} \alpha_m^2 + \frac{1}{c} v f_{m-1} - c v(\frac{1}{c} f_{m-1} + (\alpha_m^2+ \frac{1}{c}v) f_{m-1}')\\
&= f_{m-1} \alpha_m^2 + \frac{1}{c} v f_{m-1} - v f_{m-1} - c v \alpha_m^2 f_{m-1}' - v^2 f_{m-1}'\\
&= (f_{m-1} - c v f_{m-1}')\alpha_m^2 + v((\frac{1}{c}-1) f_{m-1} - v f'_{m-1})\\
&= F_{m-1} \alpha_m^2 + v \left((1/c-1) f_{m-1} - v f'_{m-1}\right).
\end{align*}
If $F_{m-1}$ is associated to $F:=(\frac{1}{c}-1) f_{m-1} - v
f'_{m-1}$ in the polynomial ring $\QQ[v, \alpha_1, \dots,
\alpha_{m-1}]$, then we would have $(\frac{1}{c}-1) F_{m-1} = F$
leading to the contradiction $(\frac{1}{c}-1) c = 1$.  Since $F_{m-1}$
is irreducible by induction, it follows that
\[
\gcd(F_{m-1}, (1/c-1) f_{m-1} - v f'_{m-1}) = 1
\]
and therefore that 
\[
F_m=F_{m-1} \alpha_m^2 + v((1/c-1) f_{m-1} - v f'_{m-1})\in
\QQ[v, \alpha_1, \dots, \alpha_{m-1}][\alpha_m]
\]
is irreducible.
\end{proof}

We recall the following version of Hilbert's irreducibility theorem
(see \cite{FJ}, Corollary 11.7) sufficient for our purposes.

\begin{theorem}[Hilbert]
\label{THM:HilbertIT}
  Let $K$ be a finite extension field of $\QQ$ and $f\in K[x, y_1,
  \dots, y_n]$ an irreducible polynomial. Then there exists
  $(\alpha_1, \dots, \alpha_n)\in \QQ^n$, such that
\[
f(x, \alpha_1, \dots, \alpha_n)\in K[x]
\]
is an irreducible polynomial.
\end{theorem}

We are now in position to show existence of the Shapley game
$\Gamma_{N,m}$.

\begin{theorem}
  For any $N,m \geq 1$ there exists a Shapley game with $N+1$
  positions each having $m$ actions for each player, such that the
  value position $N+1$ in the game is an algebraic number of degree $m^N$.
\label{THM:DegreeLowerBound}
\end{theorem}
\begin{proof}
  We shall construct the first $N$ positions as independent Shapley
  games described as above. For the base case of $N=1$, using
  Lemma~\ref{LEM:MultivariateIrreducible} we simply invoke
  Theorem~\ref{THM:HilbertIT} on the polynomial $F_m(v,\alpha_1^2,
  \dots, \alpha_m^2)$ with $c=2$, say. This gives a specialization of
  $\alpha_1,\dots,\alpha_m \in \QQ$ such that the value of the game
  $\Gamma((\alpha_1^2,\dots,\alpha_m^2),1/2)$ is an algebraic number
  $v_1$ of degree $m$.

  Now assume by induction that we have constructed $N-1$
  single-position Shapley games with values $v_1,\dots,v_{N-1}$
  together with positive integer coefficients $k_1,\dots,k_{N-1}$ such
  that $v'=k_1v_1+\dots+k_{N-1}v_{N-1}$ is an algebraic number of
  degree $m^{N-1}$. Invoke Theorem~\ref{THM:HilbertIT} on the
  polynomial $F_m(v,\alpha_1^2,\dots,\alpha_m^2)$ as before, but now
  over the extension field $\QQ(v')$. This again gives a
  specialization of $\alpha_1,\dots,\alpha_m \in \QQ$ such that the
  value of the game $\Gamma((\alpha_1^2,\dots,\alpha_m^2),1/2)$ is an
  algebraic number $v_N$ of degree $m$, but now over $\QQ(v')$. We may
  now find a positive integer $k$ such that $v'+k_Nv_N$ is an
  algebraic number of degree $m^{N-1}m=m^N$ over $\QQ$.

  Now we may construct the $N+1$ position game as follows. Let
  $K=k_1+\dots+k_N$. In position $N+1$, regardless of the players
  actions, with probability $1/2$ the game ends, and with probability
  $1/2k_i$ the play proceeds in position $i$. No payoff is
  awarded. Clearly the value of position $N+1$ is exactly
  $(k_1v_1+\dots+k_Nv_N)/2$ and is thus an algebraic number of degree $m^N$.
\end{proof}

\subsection{An explicit specialization}

Write $E_k(\alpha) = E_k(\alpha_1,\dots,\alpha_m)$ for the $k$th
elementary symmetric polynomial in $\alpha_1,\dots,a_m$ i.e.
\[
E_k(\alpha) = \sum_{1\leq i_1 < i_2 < \cdots  < i_k \leq m}{\alpha_{i_1} \cdots \alpha_{i_k}}
\]
for $1\leq k\leq m$. For notational convenience we define $E_0(\alpha)
= 1$. We have not been able to find a reference in the literature for
the following lemma. For a complete factorization of $S_k(x)$ we refer
to Lemma~\ref{lem:sym}.

\begin{lemma}\label{Lemma:IsThisKnown}
Let $S_k(x) = E_k(1, x, \dots, x^{m-1})$, where $x$ is a variable. Then
\[
\gcd(S_1(x), \dots, S_{m-1}(x)) = \Phi_m(x),
\]
where $\Phi_m$ is the $m$-th cyclotomic polynomial.
\end{lemma}
\begin{proof}
Define
\begin{align*}
f(t, x) &= (t-1)(t-x)(t-x^2) \cdots (t-x^{m-1})\\
&= t^m - S_1(x)t^{m-1} + \cdots + (-1)^{m-1} S_{m-1}(x)t +(-1)^m S_m(x).
\end{align*}
If $\xi$ is a primitive $m$-th root of unity, then $f(t, \xi) = t^m -
1$ and therefore $S_j(\xi) = 0$ for $j = 1, \dots, m-1$. If $\xi$ is
not a primitive $m$-th root of unity, then $f(t, \xi)$ has multiple
roots showing that $S_j(\xi)\neq 0$ for some $j = 1, \dots, m-1$. Thus
the greatest common divisor of $S_1(x), \dots, S_{m-1}(x)$ is the
product of $(x-\xi)$, where $\xi$ runs through the primitive $m$-th
roots of unity. This polynomial is precisely $\Phi_m$.
\end{proof}

We now derive the following formula for $F_m$ giving the coefficients
explicitly.
\begin{lemma}
\[
F_m(v) = \sum_{k=0}^m E_{m-k}(\alpha)(1-c k) (v/c)^k \enspace .
\]
\end{lemma}
\begin{proof}
First we have 
\[
f_m(v) = \prod_{i=1}^m (a_i + v/c) = \sum_{k=0}^m E_{m-k}(\alpha) (v/c)^k \enspace ,
\]
and thus
\[
f'_m(v) = \sum_{k=0}^m E_{m-k}(\alpha) k v^{k-1}(1/c)^k \enspace .
\]
We can then conclude
\[
F_m(v) = \sum_{k=0}^m E_{m-k}(\alpha) ( (v/c)^k - c v (kv^{k-1}(1/c)^k)) =  \sum_{k=0}^m E_{m-k}(\alpha) ( (1-c k)(v/c)^k) \enspace .
\]
\end{proof}

\begin{lemma}
  The polynomial 
  \begin{equation}
    F(v) = \sum_{k=0}^m E_{m-k}(\alpha)(1 - c k) c^{m-k} v^k  = c^m  F_m(v) \enspace,
    \label{eq:F-poly}
  \end{equation}
  is irreducible for an infinite number of specializations of $\alpha$
  and $c$.
\end{lemma}
\begin{proof}
  We consider the polynomial $G(v) = v^m F(1/v)$. Obviously $F(v)$ is
  irreducible if and only if $G(v)$ is.  Moreover, we let $\alpha_i=x^{i-1}$,
  for $1\leq i \leq m$, for $x \in \ZZ_+$ to be specified in the
  sequel.  By abuse of
  notation we also denote this polynomial as $G(v)$, which is
  \begin{displaymath}
    G(v) = \sum_{k=0}^{m}{ (1-(m-k)c)c^k \cdot S_{k}(x) \cdot v^{k}}.
  \end{displaymath}

  By 
  Lemma~\ref{Lemma:IsThisKnown}
  all the coefficients of $G(v)$, except the
  leading and the trailing coefficient, have $\Phi_m(x)$ as a common
  divisor.
  Now specialize to $x = \ell m$ with $\ell \gg 0$ and $\ell\in
  \NN$. Let $p$ be a prime divisor in $\Phi_m(x)$. Then $p\nmid
  x$. There exists infinitely many $c\in \NN$, such that $p\mid 1 - m c$,
  since $p\nmid m$. By possibly replacing $c$ by $c + p$ we may assume
  that $1 - m c = b p$, where $p\nmid b$.

  With this choice of $c$, $p\nmid c$ and $p$ divides the constant
  term of $G(v)$ precisely once.  Moreover, $p$ is not a divisor
  of the leading coefficient of $G(v)$, which is $x^{m(m-1)/2} c^m$.

  We conclude using Eisenstein's criterion (Theorem~\ref{th:eisenstein})
  all but   that $G(v)$, and hence $F(v)$, is irreducible for this class of (infinite)
  specializations.
\end{proof}

\begin{lemma}
  Let $F_j(v)$ as in (\ref{eq:F-poly}), i.e.
  \begin{equation}
    F_j(v) = \sum_{k=0}^m E_{m-k}(a_{1j}, \dots, a_{mj})(1-c_j k) c_j^{m-k} v^k \enspace,
    \label{eq:F-polys}
  \end{equation}
  where $1 \leq j \leq n$.
  Let $\gamma_j$ be any root of $F_j(v)$, then 
  there is an infinite number of specializations of $a_{ij}$ and $c_j$, such that 
  \begin{displaymath}
    [\QQ(\gamma_1, \dots, \gamma_n):\QQ] = m^n.  
  \end{displaymath}
  \label{lem:tower-alg}
\end{lemma}
\begin{proof}
  We consider
  the specialization $a_{ij} = x^{i-1}$, where $1\leq i \leq m$, 
  $1\leq j \leq n$, for a $x \in \ZZ_+$ to be specified in the sequel.

  As before, we let $S_k(m) = E_k(1, x, \dots, x^{m-1})$, and we perform
  the transformations $G_j(v) = v^m F_j(1/v)$, where $1 \leq j \leq n$,
  and we obtain the polynomials
  \begin{displaymath}
    G_j(v) = \sum_{k=0}^{m}{ (1- (m-k)c_j) \, c_j^k \cdot S_{k}(m) \cdot v^{k}}.
  \end{displaymath}

  We pick a $x \in \ZZ_+$ so that $\Phi_m(x)$ has at least $n$ distinct
  prime factors, $p_1, \dots, p_n$, that are relative prime to $m$.
  For such a procedure we refer to Lemma~\ref{lem:get-primes}.
  For $1 \leq j \leq n$, we choose $c_j$ 
  so that the equation $1-m c_j = b_j p_j$ is satisfied for
  an integer $b_j$, and $p_j$ is not a divisor of $b_j$.

  All, but the leading and trailing, coefficients of $G_j$
  have $\Phi_m(x)$ as their common GCD, according to Lemma~\ref{Lemma:IsThisKnown},
  and hence they are $0 \mod p_j$, for  $1 \leq j \leq n$.

  To summarize, for the $n$ primes, $p_j$, it holds:
  \begin{itemize}
  \item None of them divides any of the leading coefficients of $G_j$.
  \item For each $j$, $p_j$ divides the constant term of
    $G_j(v)$, $p_j^2$ does not,
    and $p_j$ does not divide any of the constant term of the other polynomials.
  \item For all $G_j$, all the coefficients but the leading and the
    constant term, are divided by $p_j$.
  \end{itemize}
  Hence, according to Theorem~\ref{th:gen_eisen}, if $\gamma_j$ is a root of
  $G_j$, then 
  \begin{displaymath}
    [\QQ(\gamma_1, \dots, \gamma_n):\QQ] = m^n.  
  \end{displaymath}
\end{proof}
\begin{lemma}
  Let $F_j(v)$ as in (\ref{eq:F-polys}), and 
  let $\gamma_j$ be any root of $F_j(v)$, then 
  there is an infinite number of specializations of $a_{ij}$ and $c_j$, such that 
  for all but a finite number of $k_j \in \QQ$, it holds 
  \begin{displaymath}
    [\QQ(\gamma_1 + k_2 \gamma_2 + \cdots + k_n \gamma_n):\QQ] = m^n \enspace, 
  \end{displaymath}
  where $1 \leq j \leq n$.
\label{LEM:LinearCombinationExtension}
\end{lemma}
\begin{proof}
  The existence of $k_i$ is guaranteed from the existence of primitive
  element \cite{Wae}.
  That is for all but a finite number of values  of $k_j \in \QQ$ it holds 
  $\QQ(\gamma_1, \dots, \gamma_n) = \QQ(\gamma_1 + k_2 \gamma_2 + \cdots + k_n \gamma_n)$,
  and from Lemma~\ref{lem:tower-alg} we conclude for the degree.
  
  To find explicit values for $k_i$ we modify slightly the proof of the existence of
  primitive element \cite{Wae}.  Let $\gamma_{ji}$ be all the roots of
  $F_j(v)$, where $1 \leq i \leq m$. It is without loss of generality
  to assume that $\gamma_i = \gamma_{j1}$.

  Let $\beta_2 = \gamma_1 + k_2\gamma_2$.
  For $\QQ(\beta_2) = \QQ(\gamma_1, \gamma_2)$ to hold, it should be 
  \begin{displaymath}
    k_2 \not= \frac{\gamma_{11} -\gamma_{1i}}{\gamma_{2 \ell} -\gamma_{21}} \enspace,
  \end{displaymath}
  for all $1\leq i \leq m$ and $1<\ell\leq m$, and hence 
  there are at most $(m-1)m$ forbidden values for $k_2$.
  This means that there is at least one positive integer between $0$ and $m^2$, 
  that we can assign $k_2$ to, so that
  $\QQ(\gamma_1 + k_2 \gamma_2) = \QQ(\beta_2) = \QQ(\gamma_1, \gamma_2)$.

  If we let $\beta_3 = \beta_2 + k_3 \gamma_3 = \gamma_1 + k_2 \gamma_2 + k_3\gamma_3$, then 
  for $\QQ(\beta_3) = \QQ(\gamma_1, \gamma_2, \gamma_3)$ to hold, it should be 
  \begin{displaymath}
    k_3 \not= \frac{\beta_{21} -\beta_{2i}}{\gamma_{3 \ell} -\gamma_{31}}=
    \frac{(\gamma_{11} -\gamma_{1,i_1}) + k_2(\gamma_{21} -\gamma_{2,i_2})}{\gamma_{3 \ell} -\gamma_{31}}
    \enspace,
  \end{displaymath}
  for all $1\leq i \leq m^2$, $1<i_1\leq m$, $1<i_2\leq m$ and
  $1<\ell\leq m$.  Hence there are at most $(m-1) m^2)$ forbidden
  values for $k_3$, and so there at least two integers between 0 and
  $(m-1)m + (m-1)m^2 = (m-1)^2 m$, that $k_2$ and $k_3$ could be
  assign to, so that $\QQ(\beta_3) = \QQ(\gamma_1, \gamma_2, \gamma_3)$.
  
  We continue similarly, and eventually we let 
  \begin{displaymath}
    \beta = \beta_n = \beta_{n-1} + k_n \gamma_n =\gamma_1 + k_2 \gamma_2 + \cdots + k_n \gamma_n
    \enspace.
 \end{displaymath}
 We consider
 \begin{displaymath}
   k_n \not= \frac{\beta_{n-1,1} -\beta_{n-1,i}}{\gamma_{n \ell} -\gamma_{n1}}=
    \frac{(\gamma_{11} -\gamma_{1,i_1}) + k_2(\gamma_{21} -\gamma_{2,i_2})
    + \cdots + k_{n-1}(\gamma_{n-1,1} -\gamma_{n-1,i_{n-1}}) }{\gamma_{n \ell} -\gamma_{n1}}
    \enspace,
  \end{displaymath}
  for all $1\leq i \leq m^{n-1}$,
  $1 \leq i_{\ell} \leq m$, and $1<\ell\leq m$.
  There are at $m^{n-1}(m-1)$ forbidden values for $k_n$.
  
  Overall, there is at least $n-1$ integers between
  $0$ and $m(m^{n-1} -1) \sim m^n$ 
  that $k_2, \dots, k_n$, could be assigned to, so that
  \begin{displaymath}
    \QQ(\gamma_1, \dots, \gamma_n) = \QQ(\gamma_1 + k_2 \gamma_2 + \cdots + k_n \gamma_n) \enspace .
  \end{displaymath}
  Using the previous lemma we conclude for the degree.
\end{proof}

Now combining Lemma~\ref{lem:tower-alg} and
Lemma~\ref{LEM:LinearCombinationExtension} we can immediately turn the
proof of Theorem~\ref{THM:DegreeLowerBound} into an explicit proof of
existence.

\subsubsection{Auxiliary results}


A similar lemma appears in \cite{s-b2-1999}.
\begin{lemma}
  \label{lem:sym}
  Let $E_k$ be the elementary symmetric polynomials in $n$ variables
  $a_1, \dots,  a_n$, where $0 \leq k \leq n$, that is 
  $E_k(a_1, \dots, a_n) = \sum_{1\leq i_1 < i_2 < \cdots  < i_k \leq n}{a_{i_1} \cdots a_{i_k}}$.

  Let $S_k(n) = E_k(1, x, \dots, x^{n-1})$, 
  where $1\leq k \leq n$, then it holds that
  \begin{displaymath}
    \begin{aligned}
      S_k(n) & = x^{k(k-1)/2} \prod_{\ell=1}^{k} (x^{n-\ell+1} -1) / \Phi_{\ell}^{\floor{k/\ell}}(x) \\
&=x^{k(k-1)/2} \prod_{i_1|n}{\Phi_{i_1}(x)} \prod_{i_2|(n-1)}{\Phi_{i_2}(x)} \cdots \prod_{i_k|(n-k+1)}{\Phi_{i_k}(x)} /
  \Paren{ \prod_{\ell=1}^{k} \Phi_{\ell}^{\floor{k/\ell}}(x) },
    \end
{aligned}
  \end{displaymath}
  where $\Phi_{\ell} = \Phi_{\ell}(x)$ is the $\ell$-th cyclotomic polynomial.
\end{lemma}

\begin{proof}
  We prove the formula using double induction.

  Evidently the formula holds for $S_1(1)$, 
  and we can easily prove that it holds for $S_1(n)$, for every $n$.  
  It also holds for $S_k(k)$ for all $k$.

  For the definition of the elementary symmetric polynomials it holds that 
  $E_k(a_1, \dots, a_n) = E_k(a_1, \dots, a_{n-1}) + a_n E_{k-1}(a_1, \dots, a_{n-1})$, 
  and hence
  $S_k(n) = S_k(n-1) + x^{n-1} S_{k-}(n-1)$.
  We assume that the formula holds for $S_k(n-1)$ and $S_{k-1}(n-1)$
  and we prove that it holds for $S_k(n)$.

  \begin{displaymath}
    \begin{aligned}
      S_k(n) &= S_k(n-1) + x^{n-1}S_{k-1}(n-1) \\
      &=x^{k(k-1)/2} \prod_{\lambda=1}^{k} \frac{x^{n-\lambda}-1}{\Phi_{\lambda}^{\floor{k/\lambda}}} 
      + x^{n-1} \cdot x^{(k-1)(k-2)/2} \prod_{\mu=1}^{k-1} \frac{x^{n-\mu}-1}{\Phi_{\mu}^{\floor{(k-1)/\mu}}} \\
      &=x^{k(k-1)/2} \frac{\prod_{\lambda=1}^{k-1}(x^{n-\lambda}-1)}{\prod_{\lambda=1}^{k}\Phi_{\lambda}^{\floor{k/\lambda}}}
      (x^{n-k} - 1 + x^{n-k} \prod_{\mu|k}{\Phi_{\mu}}) \\
      &=x^{k(k-1)/2} \frac{\prod_{\lambda=1}^{k-1}(x^{n-\lambda}-1)}{\prod_{\lambda=1}^{k}\Phi_{\lambda}^{\floor{k/\lambda}}}
      (x^{n-k} - 1 + x^{n-k} (x^k -1)) \\
      &= x^{k(k-1)/2} \prod_{\lambda=1}^{k} \frac{(x^{n-\lambda+1} -1)}{\Phi_{\lambda}^{\floor{k/\lambda}}}.
    \end{aligned}
  \end{displaymath}
  
  The formula follows if we also consider that 
  $x^{n}-1 = \prod_{\ell|n}\Phi_{\ell}$.
\end{proof}

\begin{theorem}[Eisenstein's criterion]
  \label{th:eisenstein}
  Let $f(x) = a_n x^n + \cdots + a_1 x + a_0$ be a polynomial with
  integer coefficients.
  Let $p$ be a prime such that (i) $p$ divides each $a_i$ for 
  $1 \leq i < n$, (ii) $p$ does not divide $a_n$, and  (iii) $p^2$ does
  not divide $a_0$, then $f$ is irreducible over the rational numbers. 
\end{theorem}

\begin{theorem}[Generalized Eisenstein's criterion]
  \cite{j-jnt-1990}
  \label{th:gen_eisen}
  Let 
  \begin{displaymath}
    f_i(x) = x^{n_i} + a_{i,1}x^{n_i-1} + \cdots + a_{i, n_i},
  \end{displaymath}
  where $1 \leq i \leq s$ and all the coefficients of all the polynomials belong to $O$.

  If there exists non-archimedean valuations $v_1, v_2, \dots, v_s$ of $K$ such that 
  $t(v_1) = p_1, \dots, t(v_s)=p_s$ are distinct primes, and that
  \begin{displaymath}
    v_i(a_{i, n_i}) = 1, \quad v_i(a_{j, n_j}) = 0, \text{ and } v_i(a_{k,r}) \geq 1,
  \end{displaymath}
  where $1\leq i,j,k\leq s$, $i \not= j$, $1\leq r \leq n_k-1$, then,
  for any choice of the roots of $f_i(x)$, say $\gamma_i$, $1 \leq i
  \leq s$, we have
  \begin{displaymath}
    [K(\gamma_1, \dots, \gamma_s):K] = n_1 n_2 \cdots n_s.
  \end{displaymath}
\end{theorem}

\begin{remark}[A note on the leading coefficient]
  Theorem~\ref{th:gen_eisen} is a generalization of Eisenstein's
  criterion.  It assumes that all the polynomials are monic. However,
  it is without loss of generality to assume that the corresponding
  primes do not divide the leading coefficients.  This is so because
  we can transform a non-monic polynomial to a monic one, such
  that the theorem holds, as follows:
  Given a polynomial
  \begin{displaymath}
    g(x) = a_n x^n + a_{n-1}x^{n-1} + \cdots + a_1 x  + a_0,
  \end{displaymath}
  where $a_n \not=1$, we multiply all the coefficients by $a_n^{n-1}$ then 
  \begin{displaymath}
    g_1(x) 
    = a_n^{n} x^n + a_{n-1}a_n^{n-1}x^{n-1} + \cdots + a_n^{n-1}a_1 x  + a_n^{n-1}a_0.
  \end{displaymath}
  If we set $y = a_n x$, then 
  \begin{displaymath}
    g_2(y) = y^n + a_{n-1}y^{n-1} + \cdots + a_n^{n-2}a_1 y  + a_n^{n-1}a_0
    = y^n + c_{n-1} y^{n-1} + \cdots + c_1 y  + c_0
  \end{displaymath}
  where in order Eisenstein's criterion, or its generalization, to hold,
  it suffices a prime $p$ not to divide the leading coefficient of
  $g$, $a_{n}$. 
  If the roots of $g_2$ are $\beta_i$, then the roots of $g$ are
  $\gamma_i = \beta_i/a_n$.
\end{remark}

\begin{lemma}
  \label{lem:get-primes}
  Let $f\in \ZZ[x]$ be a non-constant polynomial and $a$ an integer,
  such that $f(a)$ has the prime divisors $p_1, \dots, p_k$ with
  $k\geq 1$.  Then there exists an integer $b$, such that $f(b)$ is
  divisible by at least $k+1$ primes.
\end{lemma}
\begin{proof}
We consider the polynomial
\[
g(x) := f(f(a)^2 x + a) = f(a) + f(a)^2 x h(x) = f(a)(1 + f(a) x h(x)),
\]
where $h(x)\in \ZZ[x]$ is non-zero. Notice that $g(x)$ is divisible by
$p_1, \dots, p_k$ for every $x\in \ZZ$.
Now the result follows, since a
prime dividing $1 + f(a) x h(x)$ for $x\in \ZZ$ cannot be among 
$p_1, \dots, p_k$.
\end{proof}

\section{Upper bounds for value and strategy iteration for concurrent reachability games }
\label{sec-si}

In this section we explain how the techniques of Section~\ref{sec-upp}
as used for Everett games, also yields an improved analysis of the
strategy improvement algorithm for concurrent reachability games.

Let $\Gamma$ be an Everett game, with $N$ positions. Assume that in
position $k$, the two players have $m_k \leq m$ and $n_k \leq m$ actions
available. Assume further that all payoffs and probabilities in
$\Gamma$ are rational numbers with numerators and denominators of
bitsize at most $\tau$. Further, let $\sigma$ be a fixed positive integer.

From Lemma~\ref{LEM:QFreeC1C2} we get the following statement.
\begin{lemma}
\label{LEM:C1C2point}
  There is a quantifier free formula with $2N$ free variables $v_1$
  and $v_2$ that expresses $v_1 \in C_1(\Gamma), v_2 \in
  C_2(\Gamma)$, and $\norm{v_1-v_2}^2 \leq 2^{-\sigma}$.
  
  The formula uses at most $(2N+1)+2(m+2)\sum_{k=1}^N
  \binom{n_k+m_k}{m_k}$ different polynomials, each of degree at most
  $m+2$ and having coefficients of bitsize at most $\max(\sigma,2(N+1)(m+2)\tau)$,
  where $m = \max_{k=1}^N\left(\min(n_k,m_k)\right)$.
\end{lemma}

\begin{theorem}
  \label{th:e-strategy-prob}
  Let $\Gamma$ and $\sigma$ be as above. Let $\epsilon=2^{-\sigma}$. Then
  there exists $\epsilon$-optimal strategy of $\Gamma$ where each probability
  is a real algebraic number, defined by a polynomial of
  degree $m^{O(N)}$ and maximum coefficient bitsize $\max(\sigma,\tau) m^{O(N)}$.
\end{theorem}
\begin{proof}
  We use Theorem 13.11 of \cite{BasuPollackRoy2006} to find a
  univariate representation of the pair $(v_1,v_2)$ satisfying the
  formula from Lemma~\ref{LEM:C1C2point}. That is we have polynomials
  $f,g_0,\dots,g_{2N}$, with $f$ and $g_0$ coprime, such that the
  points $(v_1,v_2)$ are given as
  $(g_1(t)/g_0(t),\dots,g_{2N}(t)/g_0(t))$, where $t$ is a root of
  $f$. These polynomials are of degree $m^{O(N)}$ and their maximum
  coefficient bitsize is $\max(\sigma,\tau) m^{O(N)}$.
  
  Now consider the matrix games $A^k(v_1)$ for all positions $k$. We
  find optimal strategies $p^1,\dots,p^N$ that correspond to basic
  feasible solutions of the linear program LP~(\ref{EQ:MatrixLP}).
  Notice that the elements of these matrix games are rational
  polynomial functions in $g_0,\dots,g_{N}$. By Lemma~\ref{LEM:LPbfs}
  we have $p^k_i =\det((M^{A^k}_{B^k})_i)/\det(M^{A^k}_{B^k})$ for
  some potential basis sets $B^1,\dots,B^k$. Using
  Lemma~\ref{PROP:DetPolyBound}, each $p^k_i$ is a rational polynomial
  function in $g_0,\dots,g_{N}$ of degree $m^{O(N)}$ and maximum
  coefficient bitsize $\max(\sigma,\tau) m^{O(N)}$. Substituting the
  root $t$ of $f$ using Lemma~\ref{LEM:convert} we obtain the
  statement.
\end{proof}

Using Lemma~\ref{lem:uni-bounds} we deduce:
\begin{corollary}
  An Everett game with coefficient bitsize bounded by $\tau$ has
  a $2^{-\sigma}$ optimal strategy where the probabilities
  are either zero or bounded from below by $2^{-\max(\sigma,\tau) m^{O(N)}}$.
\end{corollary}

We now apply Lemma 3 of Hansen, Ibsen-Jensen and Miltersen \cite{CRGIteration} and conclude that
value iteration and strategy iteration on a deterministic concurrent
reachability game (where $\tau = O(1)$) will compute an
$\epsilon$-optimal
strategy after at most $(\frac{1}{\epsilon})^{m^{O(N)}}$ iterations.
This matches the lower bound obtained by Hansen, Ibsen-Jensen and Miltersen \cite{CRGIteration}.


\bibliographystyle{plain}
\bibliography{DSGAlgorithm}
\end{document}